\documentclass[]{llncs}

\usepackage[misc,geometry]{ifsym}  
\usepackage{xspace}
\usepackage[english]{babel}
\usepackage{amsmath}
\usepackage{amsopn}
\usepackage{latexsym}
\usepackage{textcomp}
\usepackage{subfigure}
\usepackage{color}
\usepackage{enumerate}
\usepackage[inline]{enumitem}
\usepackage{tabularx}
\usepackage[plainpages=false]{hyperref}
\usepackage[figure,table]{hypcap}
\usepackage{multirow}
\usepackage[disable]{todonotes}
\usepackage{lineno}
\usepackage{graphicx}
\usepackage{authblk}
\usepackage{amssymb,amsfonts}
\usepackage{thm-restate}
\usepackage{wrapfig}

\usepackage{nopageno}

\graphicspath{{figures/}}

\newcounter{casecounter}
\newcounter{subcasecounter}
\newcounter{subsubcasecounter}
\makeatletter
\newcommand{\ccase}[2][]{%
	\stepcounter{casecounter}%
	\setcounter{subcasecounter}{0}%
	\protected@write \@auxout {}{\string \newlabel {#2}{{#1\thecasecounter}{\thepage}{#1\thecasecounter}{#2}{}} }%
	\hypertarget{#2}{\noindent\textbf{Case #1\thecasecounter.}}
}

\newcommand{\subcase}[2][]{%
	\stepcounter{subcasecounter}%
	\setcounter{subsubcasecounter}{0}%
	\protected@write \@auxout {}{\string \newlabel {#2}{{#1\thecasecounter.\thesubcasecounter}{\thepage}{#1\thecasecounter.\thesubcasecounter}{#2}{}} }%
	\hypertarget{#2}{\noindent\textbf{Case #1\thecasecounter.\thesubcasecounter.}}
}

\newcommand{\subsubcase}[2][]{%
	\stepcounter{subsubcasecounter}%
	\protected@write \@auxout {}{\string \newlabel {#2}{{#1\thecasecounter.\thesubcasecounter.\thesubsubcasecounter}{\thepage}{#1\thecasecounter.\thesubcasecounter.\thesubsubcasecounter}{#2}{}} }%
	\hypertarget{#2}{\noindent\textbf{Case #1\thecasecounter.\thesubcasecounter.\thesubsubcasecounter.}}
}
\makeatother

\newcommand{\X}{\raisebox{-2pt}{\includegraphics[scale=0.45]{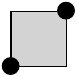}}}
\newcommand{\D}{\raisebox{-2pt}{\includegraphics[scale=0.45]{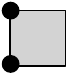}}\hspace{1pt}}
\renewcommand{\L}{\raisebox{-2pt}{\includegraphics[scale=0.45]{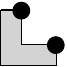}}}
\newcommand{\C}{\raisebox{-2pt}{\includegraphics[scale=0.45]{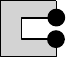}}}
\newcommand{\x}{\raisebox{-2pt}{\includegraphics[scale=0.25]{figures/X.pdf}}}
\renewcommand{\d}{\raisebox{-2pt}{\includegraphics[scale=0.25]{figures/D.pdf}}}

\newcommand{\RN}{\textsf{NoBend-Alg}\xspace}
\newcommand{\RNN}{\textsf{MinBendCubic-Alg}\xspace}
\newcommand{\R}{\textsf{Rset-Alg}\xspace}
\newcommand{\myparagraph}[1]{\smallskip\noindent\textbf{\boldmath #1}}

\newcommand{\skel}{\mathrm{skel}\xspace}
\newcommand{\rect}{\overline}

\begin{document}
	\title{Bend-minimum Orthogonal Drawings in Quadratic Time\thanks{Research supported in part by the project: ``Algoritmi e sistemi di analisi visuale di reti complesse e di grandi dimensioni - Ricerca di Base 2018, Dipartimento di Ingegneria dell'Universit\`a degli Studi di Perugia'' and by MIUR project ``MODE -- MOrphing graph Drawings Efficiently'', prot.~20157EFM5C\_001.}}
	\author{Walter Didimo\inst{1}$^{\textrm{\Letter}}$,
		Giuseppe Liotta\inst{1},
		Maurizio Patrignani\inst{2}
	}
	
	\date{}
	
	\institute{
		Universit\`a degli Studi di Perugia, Italy\\
		\email {\{walter.didimo,giuseppe.liotta\}@unipg.it}
		\and
		Universit\`a Roma Tre, Italy\\
		\email {maurizio.patrignani@uniroma3.it}
	}

	\maketitle
	

\begin{abstract}
Let $G$ be a planar $3$-graph (i.e., a planar graph with vertex degree at most three) with $n$ vertices. We present the first $O(n^2)$-time algorithm that computes a planar orthogonal drawing of $G$ with the minimum number of bends in the variable embedding setting. If either a distinguished edge or a distinguished vertex of $G$ is constrained to be on the external face, a bend-minimum orthogonal drawing of $G$ that respects this constraint can be computed in $O(n)$ time. Different from previous approaches, our algorithm does not use minimum cost flow models and computes drawings where every edge has at most two bends.
%
%
\end{abstract}

\section{Introduction}\label{se:intro}


 A pioneering paper by Storer~\cite{DBLP:conf/stoc/Storer80} asks whether a crossing-free orthogonal drawing with the minimum number of bends can be computed in polynomial time. The question posed by Storer is in the fixed embedding setting, i.e., the input is a plane 4-graph (an embedded planar graph with vertex degree at most four) and the wanted output is an embedding-preserving orthogonal drawing with the minimum number of bends. Tamassia~\cite{DBLP:journals/siamcomp/Tamassia87} answers Storer's question in the affirmative by describing an $O(n^2 \log n)$-time algorithm. The key idea of Tamassia's result is the equivalence between the bend minimization problem and the problem of computing a min-cost flow  on a suitable network.  To date, the most efficient known solution of the bend-minimization problem for orthogonal drawings in the fixed embedding setting is due to Cornelsen and  Karrenbauer~\cite{DBLP:journals/jgaa/CornelsenK12}, who show a novel technique to compute a min-cost flow on an uncapacitated network and apply this technique to Tamassia's model achieving $O(n^{\frac{3}{2}})$-time complexity.


\begin{figure}[t]
	\centering
	\hfill
	\subfigure[]{\label{fi:intro-a}\includegraphics[width=0.18\columnwidth]{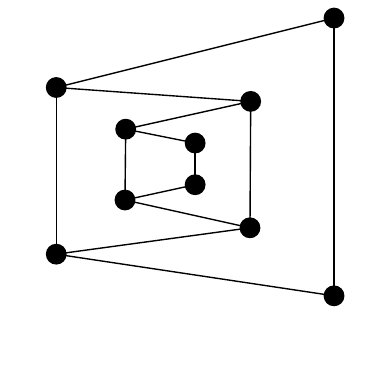}}
	\hfill
	\subfigure[]{\label{fi:intro-b}\includegraphics[width=0.18\columnwidth]{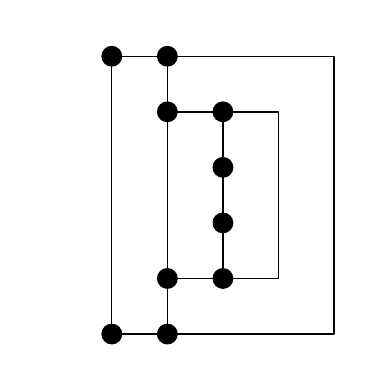}}
	\hfill
	\subfigure[]{\label{fi:intro-c}\includegraphics[width=0.18\columnwidth]{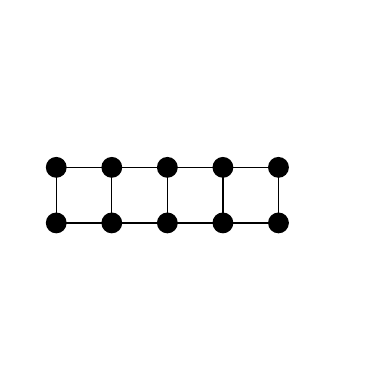}}\hfill~
	\caption{(a) A planar embedded 3-graph $G$. (b) An embedding-preserving bend-minimum orthogonal drawing of $G$. (c) A bend-minimum orthogonal drawing of $G$.}\label{fi:intro}
\end{figure}


A different level of complexity for the bend minimization problem is encountered in the variable embedding setting, that is when the algorithm is asked to find a bend-minimum solution over all planar embeddings of the graph. For example, the orthogonal drawing of Fig.~\ref{fi:intro-c} has a different planar embedding than the graph of Fig.~\ref{fi:intro-a} and it has no bends, while the drawing of Fig.~\ref{fi:intro-b} preserves the embedding but it is suboptimal in terms of bends.

Garg and Tamassia~\cite{DBLP:journals/siamcomp/GargT01} prove that the bend-minimization problem for orthogonal drawings is NP-complete for planar 4-graphs, while Di Battista et al.~\cite{DBLP:journals/siamcomp/BattistaLV98} show that it can be solved in $O(n^5 \log n)$ time for planar 3-graphs. Generalizations of the problem in the variable embedding setting where edges have some flexibility (i.e., they can bend a few times without cost for the optimization function) have also been the subject of recent studies by Bl{\"{a}}sius et al.~\cite{DBLP:journals/talg/BlasiusRW16}.

Improving the $O(n^5 \log n)$ time complexity of the algorithm by Di Battista et al.~\cite{DBLP:journals/siamcomp/BattistaLV98} has been an elusive open problem for more than a decade (see, e.g.,~\cite{DBLP:conf/gd/BrandenburgEGKLM03}), until a paper by Chang and Yen~\cite{DBLP:conf/compgeom/ChangY17} has shown how to compute a bend-minimum orthogonal drawing of a planar 3-graph in the variable embedding setting in $\tilde{O}(n^\frac{17}{7})$ time, which can be read as $O(n^\frac{17}{7}\log^k n)$ time for a positive constant $k$.

Similar to~\cite{DBLP:journals/siamcomp/BattistaLV98}, the approach in~\cite{DBLP:conf/compgeom/ChangY17} uses an SPQR-tree to explore all planar embeddings of a planar 3-graph and combines partial solutions associated with the nodes of this tree to compute a bend-minimum drawing. Both in~\cite{DBLP:journals/siamcomp/BattistaLV98} and in~\cite{DBLP:conf/compgeom/ChangY17}, the computationally most expensive task is computing min-cost flows on suitable variants of Tamassia's network. However, Chang and Yen elegantly prove that a simplified flow network where all edges have unit capacity can be adopted to execute this task. This, combined with a recent result~\cite{DBLP:conf/focs/CohenMTV17} about min-cost flows on unit-capacity networks, yields the improved time complexity.


\myparagraph{Contribution and outline.} This paper provides new algorithms to compute bend-minimum orthogonal drawings of planar 3-graphs, which improve the time complexity of the state-of-the-art solution. We prove the following.

\begin{restatable}{theorem}{main}\label{th:main}
Let $G$ be an $n$-vertex planar 3-graph. A bend-minimum orthogonal drawing of $G$ can be computed in $O(n^2)$ time. If either a distinguished edge or a distinguished vertex of $G$ is constrained to be on the external face, a bend-minimum orthogonal drawing of $G$ that respects the given constraint can be computed in $O(n)$ time. Furthermore, the computed drawings have at most two bends per edge, which is worst-case optimal.
\end{restatable}

As in~\cite{DBLP:journals/siamcomp/BattistaLV98} and in~\cite{DBLP:conf/compgeom/ChangY17},  the algorithmic approach of Theorem~\ref{th:main} computes a bend-minimum orthogonal drawing by visiting an SPQR-tree of the input graph. However, it does not need to compute min-cost flows at any steps of the visit, which is the fundamental difference with the previous techniques. This makes it possible to design the first quadratic-time algorithm to compute bend-minimum orthogonal drawings of planar 3-graphs in the variable embedding setting. 

The second part of the statement of Theorem~\ref{th:main} extends previous studies by Nishizeki and Zhou~\cite{DBLP:journals/siamdm/ZhouN08}, who give a first example of a linear-time algorithm in the variable embedding setting for planar 3-graphs that are partial two-trees. The bend-minimum drawings of Theorem~\ref{th:main} have at most two bends per edge, which is a desirable property for an orthogonal representation. We recall that every planar 4-graph (except the octahedron) has an orthogonal drawing with at most two bends per edge~\cite{DBLP:journals/comgeo/BiedlK98,DBLP:journals/dam/LiuMS98}, but minimizing the number of bends may require some edges with a $\Omega(n)$ bends~\cite{DBLP:journals/siamcomp/BattistaLV98,DBLP:journals/ipl/TamassiaTV91}. It is also proven that every planar 3-graph (except $K_4$) has an orthogonal drawing with at most one bend per edge~\cite{DBLP:journals/algorithmica/Kant96}, but the drawings of the algorithm in~\cite{DBLP:journals/algorithmica/Kant96} are not bend-minimum. Finally, a non-flow based algorithm having some similarities with ours is given in~\cite{DBLP:conf/gd/GargL99}; it neither computes bend-minimum drawings nor guarantees at most two bends per edge.

%

The paper is organized as follows. Preliminary definitions and results are in Sec.~\ref{se:preliminaries}. In Sec.~\ref{se:properties-ortho} we prove key properties of bend-minimum orthogonal drawings of planar 3-graphs used in our approach. Sec.~\ref{se:algorithm} describes our drawing algorithms. Open problems are in Sec.~\ref{se:conclusions}. Some (full) proofs are moved to the appendix.

\section{Preliminaries}\label{se:preliminaries}

We assume familiarity with basic definitions on graph connectivity and planarity (see Appendix~\ref{se:app-prel}). If $G$ is a graph, $V(G)$ and $E(G)$ denote the sets of vertices and edges of $G$. We consider \emph{simple} graphs, i.e., graphs with neither self-loops nor multiple edges.
The \emph{degree} of a vertex $v \in V(G)$, denoted as $\deg (v)$, is the number of its neighbors. $\Delta(G)$ denotes the maximum degree of a vertex of $G$; if $\Delta(G) \leq h$ ($h \geq 1$), $G$ is an \emph{$h$-graph}. A graph $G$ is {\em rectilinear planar} if it admits a planar drawing where each edge is either a horizontal or a vertical segment (i.e., it has no bend). Rectilinear planarity testing is NP-complete for planar $4$-graphs~\cite{DBLP:journals/siamcomp/GargT01}, but it is polynomially solvable for planar $3$-graphs~\cite{DBLP:conf/compgeom/ChangY17,DBLP:journals/siamcomp/BattistaLV98} and linear-time solvable for subdivisions of planar triconnected cubic graphs~\cite{DBLP:journals/ieicet/RahmanEN05}. By extending a result of Thomassen~\cite{Th84} on those $3$-graphs that have a rectilinear drawing with all rectangular faces, Rahman et al.~\cite{DBLP:journals/jgaa/RahmanNN03} characterize rectilinear plane $3$-graphs. For a plane graph $G$, let $C_o(G)$ be its external cycle ($C_o(G)$ is simple if $G$ is biconnected). Also, if $C$ is a simple cycle of $G$, $G(C)$ is the plane subgraph of $G$ that consists of $C$ and of the vertices and edges inside $C$. An edge $e=(u,v) \notin E(G(C))$ is a \emph{leg} of $C$ if exactly one of the vertices $u$ and $v$ belongs to $C$; such a vertex is a \emph{leg-vertex} of $C$. If $C$ has exactly $k$ legs and no edge embedded outside $C$ joins two of its vertices, $C$ is a \emph{k-legged cycle} of~$G$.

\begin{theorem}\label{th:RN03}\emph{\cite{DBLP:journals/jgaa/RahmanNN03}}
	Let $G$	be a biconnected plane $3$-graph. $G$ admits an orthogonal drawing without bends if and only if: $(i)$ $C_o(G)$ contains at least four vertices of degree 2; $(ii)$ each $2$-legged cycle contains at least two vertices of degree 2; $(iii)$ each $3$-legged cycle contains at least one vertex of degree 2.
\end{theorem}

As in~\cite{DBLP:journals/jgaa/RahmanNN03}, we call \emph{bad} any $2$-legged and any $3$-legged cycle that does not satisfy Condition $(ii)$ and $(iii)$ of Theorem~\ref{th:RN03}, respectively.

\begin{figure}[t]
	\centering
	\subfigure[]{\label{fi:spqr-tree-a}}{\includegraphics[width=0.26\columnwidth]{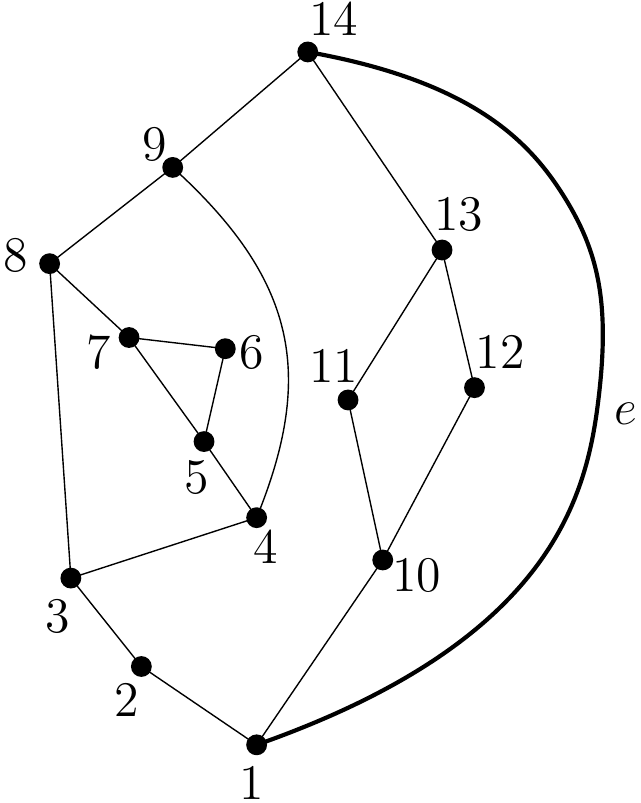}}
	\hfill
	\subfigure[]{\label{fi:spqr-tree-b}}{\includegraphics[width=0.42\columnwidth]{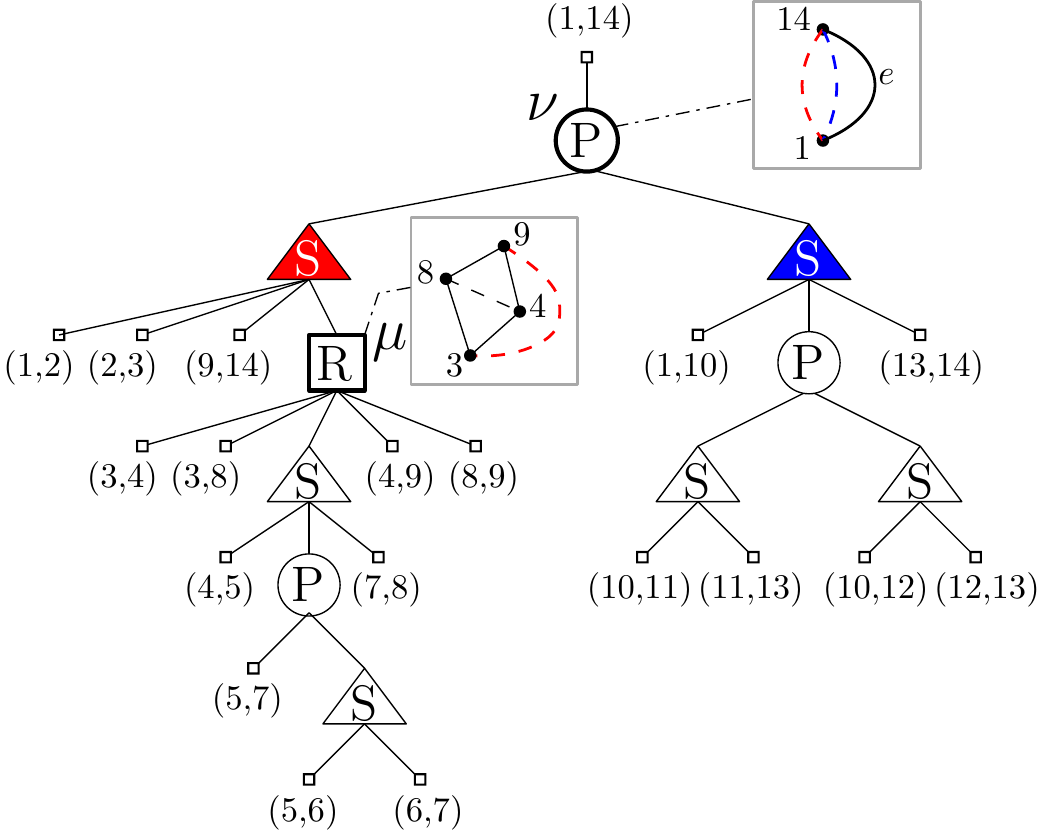}}
	\hfill
	\subfigure[]{\label{fi:spqr-tree-c}}{\includegraphics[width=0.26\columnwidth]{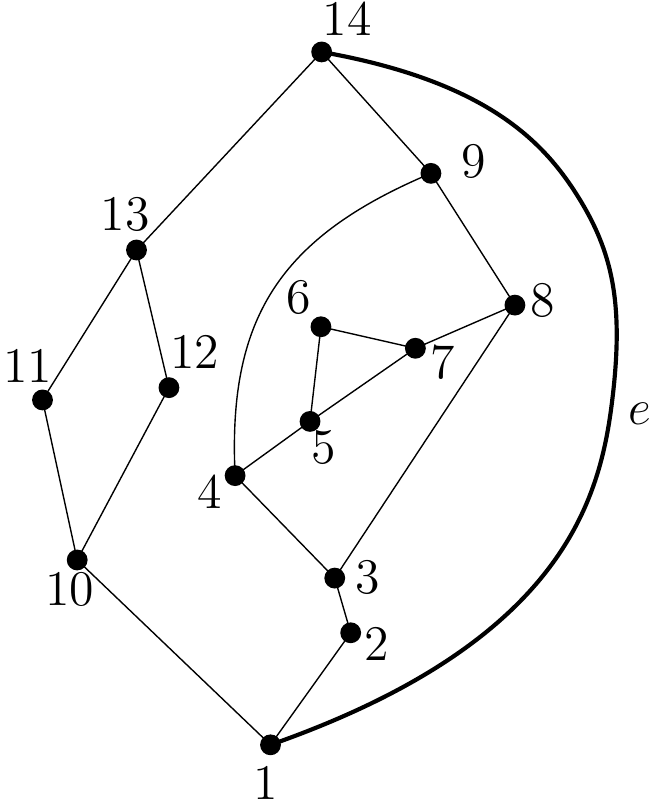}}
	\caption{(a) A plane 3-graph $G$. (b) The SPQR-tree of $G$ with respect to $e$; the skeletons of a P-node $\nu$ and of an R-node $\mu$ are shown. (c) A different embedding of $G$ obtained by changing the embedding of $\skel(\nu)$ and of $\skel(\mu)$.}\label{fi:spqr-tree}
\end{figure}

\myparagraph{SPQR-trees of Planar 3-Graphs.} Let $G$ be a biconnected graph. An \emph{SPQR-tree} $T$ of $G$ represents the decomposition of $G$ into its triconnected components and can be computed in linear time~\cite{DBLP:books/ph/BattistaETT99,DBLP:conf/gd/GutwengerM00,DBLP:journals/siamcomp/HopcroftT73}. Each triconnected component corresponds to a node $\mu$ of $T$; the triconnected component itself is called the \emph{skeleton} of $\mu$ and denoted as $\skel(\mu)$. A node $\mu$ of $T$ can be of one of the following types: $(i)$ \emph{R-node}, if $\skel(\mu)$ is a triconnected graph; $(ii)$ \emph{S-node}, if $\skel(\mu)$ is a simple cycle of length at least three; $(iii)$ \emph{P-node}, if $\skel(\mu)$ is a bundle of at least three parallel edges; $(iv)$ \emph{Q-nodes}, if it is a leaf of $T$; in this case the node represents a single edge of the graph and its skeleton consists of two parallel edges.
Note that, neither two $S$- nor two $P$-nodes are adjacent in~$T$. A \emph{virtual edge} in $\skel(\mu)$ corresponds to a tree node $\nu$ adjacent to $\mu$ in $T$. If $T$ is rooted at one of its Q-nodes $\rho$, every skeleton (except the one of $\rho$) contains exactly one virtual edge that has a counterpart in the skeleton of its parent: This virtual edge is the \emph{reference edge} of $\skel(\mu)$ and of $\mu$, and its endpoints are the \emph{poles} of $\skel(\mu)$ and of $\mu$. The edge of $G$ corresponding to the root $\rho$ of $T$ is the  \emph{reference edge} of $G$, and $T$ is the SPQR-tree of $G$ \emph{with respect to $e$}. For every node $\mu \neq \rho$ of $T$, the subtree $T_\mu$ rooted at $\mu$ induces a subgraph $G_\mu$ of $G$ called the \emph{pertinent graph} of $\mu$, which is described by $T_\mu$ in the decomposition: The edges of $G_\mu$ correspond to the Q-nodes (leaves) of $T_\mu$. Graph $G_\mu$ is also called a \emph{component} of $G$ with respect to the reference edge $e$, namely $G_\mu$ is a P-, an R-, or an S-component depending on whether $\mu$ is a P-, an R-, or an S-component, respectively.

The SPQR-tree $T$ rooted at a Q-node $\rho$ implicitly describes all planar embeddings of $G$ with the reference edge of $G$ on the external face. All such embeddings are obtained by combining the different planar embeddings of the skeletons of P- and R-nodes: For a P-node $\mu$, the different embeddings of $\skel(\mu)$ are the different permutations of its non-reference edges. If $\mu$ is an R-node, $\skel(\mu)$ has two possible planar embeddings, obtained by flipping $\skel(\mu)$ minus its reference edge at its poles. See Fig.\ref{fi:spqr-tree} for an illustration.
The child node of $\rho$ and its pertinent graph are called the \emph{root child} of $T$ and the \emph{root child component} of $G$, respectively. An \emph{inner node} of $T$ is neither the root nor the root child of $T$. The pertinent graph of an inner node is an \emph{inner component} of $G$. The next lemma gives basic properties of $T$ when $\Delta(G) \leq 3$.

\begin{restatable}{lemma}{spqrtreethreegraph}\label{se:spqr-tree-3-graph}
Let $G$ be a biconnected planar $3$-graph and let $T$ be the SPQR-tree of $G$ with respect to a reference edge $e$. The following properties hold:

\smallskip\noindent{\em \bf \textsf{T1}} Each P-node $\mu$ has exactly two children, one being an S-node and the other being an S- or a Q-node; if $\mu$ is the root child, both its children are S-nodes.

\noindent{\em \bf \textsf{T2}} Each child of an R-node is either an S-node or a Q-node.

\noindent{\em \bf \textsf{T3}} For each inner S-node $\mu$, the edges of $\skel(\mu)$ incident to the poles of $\mu$ are (real) edges of $G$. Also, there cannot be two incident virtual edges in $\skel(\mu)$.
\end{restatable}

\section{Properties of Bend-Minimum Orthogonal Representations of Planar 3-Graphs}\label{se:properties-ortho}
We prove relevant properties of bend-minimum orthogonal drawings of planar $3$-graphs that are independent of vertex and bend coordinates, but only depend on the vertex angles and edge bends. To this aim, we recall the concept of \emph{orthogonal representation}~\cite{DBLP:journals/siamcomp/Tamassia87} and define some types of ``shapes'' that we use to construct bend-minimum orthogonal representations.

\smallskip\noindent{\bf Orthogonal Representations}. Let $G$ be a plane $3$-graph. If $v \in V(G)$ and if $e_1$ and $e_2$ are two (possibly coincident) edges incident to $v$ that are consecutive in the clockwise order around $v$, we say that $a = \langle e_1,v,e_2 \rangle$ is an \emph{angle at $v$} of $G$ or simply an \emph{angle} of $G$. Let $\Gamma$ and $\Gamma'$ be two embedding-preserving orthogonal drawings of $G$. We say that $\Gamma$ and $\Gamma'$ are \emph{equivalent} if: $(i)$ For any angle $a$ of $G$, the geometric angle corresponding to $a$ is the same in $\Gamma$ and $\Gamma'$, and $(ii)$ for any edge $e=(u,v)$ of $G$, the sequence of left and right bends along $e$ moving from $u$ to $v$ is the same in $\Gamma$ and in $\Gamma'$. An \emph{orthogonal representation} $H$ of $G$ is a class of equivalent orthogonal drawings of $G$; $H$ can be described by the embedding of $G$ together with the geometric value of each angle of $G$ ($90$, $180$, $270$ degrees)\footnote{Angles of $360$ degrees only occur at 1-degree vertices; we can avoid to specify them.} and with the sequence of left and right bends along each edge. Figure~\ref{fi:ortho-a} shows a bend-minimum orthogonal representation of the graph in Fig.~\ref{fi:spqr-tree-a}.

\begin{figure}[t]
	\centering
	\subfigure[]{\label{fi:ortho-a}\includegraphics[width=0.25\columnwidth]{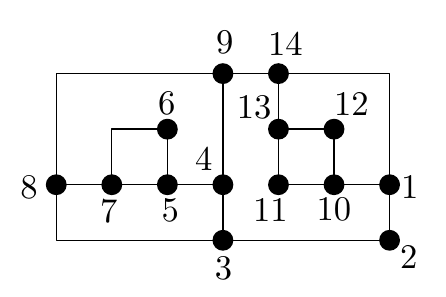}}
	\hfill
	\subfigure[]{\label{fi:ortho-b}\includegraphics[width=0.25\columnwidth]{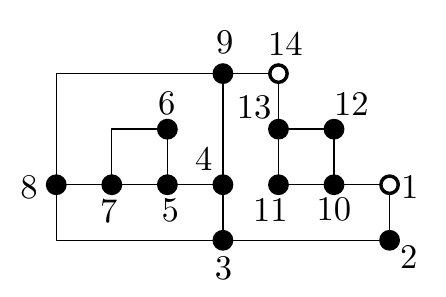}}
	\hfill
	\subfigure[]{\label{fi:ortho-c}\includegraphics[width=0.25\columnwidth]{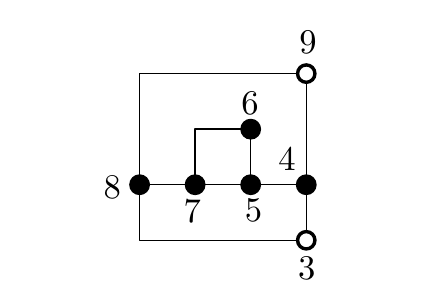}}
	\caption{(a) A bend-minimum orthogonal representation $H$ with four bends of the graph in Fig.~\ref{fi:spqr-tree-a}. (b) The component $H_\nu$, which is L-shaped; the two poles of the component are the white vertices. (c) The component $H_\mu$, which is D-shaped.}\label{fi:ortho}
\end{figure}

Let $p$ be a path between two vertices $u$ and $v$ in $H$. The \emph{turn number} of $p$ is the absolute value of the difference between the number of right and the number of left turns encountered along $p$ moving from $u$ to $v$ (or vice versa). The turn number of $p$ is denoted by $t(p)$. A turn along $p$ is caused either by a bend on an edge of $p$ or by an angle of $90/270$ degrees at a vertex of $p$. For example, $t(p)=2$ for the path $p = \langle 3, 4, 5, 6, 7 \rangle$ in the orthogonal representation of Fig.~\ref{fi:ortho-a}. We remark that if $H$ is a bend-minimum orthogonal representation, the bends along an edge, going from an end-vertex to the other, are all left or all right turns~\cite{DBLP:journals/siamcomp/Tamassia87}.

\smallskip\noindent{\bf Shapes of Orthogonal Representations}. Let $G$ be a biconnected planar $3$-graph, $T$ be the SPQR-tree of $G$ with respect to a reference edge $e \in E(G)$, and $H$ be an orthogonal representation of $G$ with $e$ on the external face. For a node $\mu$ of $T$, denote by $H_\mu$ the restriction of $H$ to a component $G_\mu$. We also call $H_\mu$ a \emph{component} of $H$. In particular, $H_\mu$ is a P-, an R-, or an S-component depending on whether $\mu$ is a P-, an R-, or an S-component, respectively. If $\mu$ is the root child of $T$, then $H_\mu$ is the \emph{root child component} of $H$. Denote by $u$ and $v$ the two poles of $\mu$ and let $p_l$ and $p_r$ be the two paths from $u$ to $v$ on the external boundary of $H_\mu$, one walking clockwise and the other walking counterclockwise. These paths are the {\em contour paths} of $H_\mu$. If $\mu$ is an S-node, $p_l$ and $p_r$ share some edges (they coincide if $H_\mu$ is just a sequence of edges). If $\mu$ is either a P- or an R-node, $p_l$ and $p_r$ are edge disjoint; in this case, we define the following \emph{shapes} for $H_\mu$, depending on $t(p_l)$ and $t(p_r)$ and where the poles are external corners:

\smallskip
\noindent$-$ $H_\mu$ is \emph{C-shaped}, or \emph{\C-shaped}, if $t(p_l)=4$ and $t(p_r)=2$, or vice versa;

\noindent$-$ $H_\mu$ is \emph{D-shaped}, or \emph{\D-shaped}, if $t(p_l)=0$ and $t(p_r)=2$, or vice versa;

\noindent$-$ $H_\mu$ is \emph{L-shaped}, or \emph{\L-shaped}, if $t(p_l)=3$ and $t(p_r)=1$, or vice versa;

\noindent$-$ $H_\mu$ is \emph{X-shaped}, or \emph{\X-shaped}, if $t(p_l)=t(p_r)=1$.

\smallskip
For example, $H_\nu$ in Fig.~\ref{fi:ortho-b} is \L-shaped, while $H_\mu$ in Fig.~\ref{fi:ortho-c} is \D-shaped. Concerning S-components, the following lemma rephrases a result in~\cite[Lemma 4.1]{DBLP:journals/siamcomp/BattistaLV98}, and it is also an easy consequence of Property \textsf{T3} in Lemma~\ref{se:spqr-tree-3-graph}.

\begin{lemma}\label{le:k-spiral}
Let $H_\mu$ be an inner S-component with poles $u$ and $v$ and let $p_1$ and $p_2$ be any two paths connecting $u$ and $v$ in $H_\mu$. Then $t(p_1)=t(p_2)$.
\end{lemma}

Based on Lemma~\ref{le:k-spiral}, we describe the shape of an inner S-component $H_\mu$ in terms of the turn number of any path $p$ between its two poles: We say that $H_\mu$ is \emph{$k$-spiral} and has \emph{spirality} $k$ if $t(p)=k$. The notion of spirality of an orthogonal component was introduced in~\cite{DBLP:journals/siamcomp/BattistaLV98}. Differently from~\cite{DBLP:journals/siamcomp/BattistaLV98}, we restrict the definition of spirality to inner S-components and we always consider absolute values, instead of both positive and negative values depending on whether the left turns are more or fewer than the right turns. For instance, in the representation of Fig.~\ref{fi:ortho-a} the two series with poles $\{1,14\}$ (the two filled S-nodes in Fig.~\ref{fi:spqr-tree-b}) have spirality three and one, respectively; the series with poles $\{4,8\}$ (child of the R-node) has spirality zero, while the series with poles $\{5,7\}$ has spirality two.
%

We now give a key result that claims the existence of a bend-minimum orthogonal representation with specific properties for any biconnected planar $3$-graph. This result will be used to design our drawing algorithm. Given an orthogonal representation $H$, we denote by $\rect{H}$ the orthogonal representation obtained from $H$ by replacing each bend with a dummy vertex: $\rect{H}$ is the \emph{rectilinear image} of $H$; a dummy vertex in $\rect{H}$ is a \emph{bend vertex}. Also, if $w$ is a degree-$2$ vertex with neighbors $u$ and $v$, \emph{smoothing} $w$ is the reverse operation of an edge subdivision, i.e., it replaces the two edges $(u,w)$ and $(w,v)$ with the single edge $(u,v)$.


\begin{restatable}{lemma}{shapes}\label{le:shapes}
A biconnected planar $3$-graph $G$ with a distinguished edge $e$
has a bend-minimum orthogonal representation $H$ with $e$ on the external face such~that:

\smallskip\noindent{\em \bf \textsf{O1}} Every edge of $H$ has at most two bends, which is worst-case optimal.

\noindent{\em \bf \textsf{O2}} Every inner P-component or R-component of $H$ is either \X- or \D-shaped.

\noindent{\em \bf \textsf{O3}} Every inner S-component of $H$ has spirality at most four.
\end{restatable}
\begin{proof}[sketch]
We prove in three steps the existence of a bend-minimum orthogonal representation $H$ that satisfies~\textsf{O1-O3}. We start by a bend-minimum orthogonal representation of $G$ with $e$ on the external face, and in the first step we prove that it either satisfies \textsf{O1} or it can be locally modified, without changing its planar embedding, so to satisfy \textsf{O1}. In the second step, we prove that from the orthogonal representation obtained in the first step we can derive a new orthogonal representation (still with same embedding) that satisfies \textsf{O2} in addition to \textsf{O1}. Finally, we prove that this last representation also satisfies \textsf{O3}. 

\smallskip\noindent{\bf Step~1: Property~\textsf{O1}}. Suppose that $H$ is a bend-minimum orthogonal representation of $G$ with $e$ on the external face and having an edge $g$ (possibly $g = e$) with at least three bends. Let $\rect{H}$ be the rectilinear image of $H$, and let $\rect{G}$ be the plane graph underlying $\rect{H}$. Since $\rect{H}$ has no bend, $\rect{G}$ satisfies Conditions~$(i)-(iii)$ of Theorem~\ref{th:RN03}. Let $v_1$, $v_2$, $v_3$ be three bend vertices in $\rect{H}$ that correspond to three bends of $g$ in $H$.
Assume first that $g$ is an internal edge of $G$ and let $\rect{G'}$ be the plane graph obtained from $\rect{G}$ by smoothing $v_1$. We claim that $\rect{G'}$ still satisfies Conditions~$(i)-(iii)$ of Theorem~\ref{th:RN03}. Indeed, if this is not the case, there must be a bad cycle in $\rect{G'}$ that contains both $v_2$ and $v_3$. This is a contradiction, because no bad cycle can contain two vertices of degree two. Hence, there exists an (embedding-preserving) representation $\rect{H'}$ of $\rect{G'}$ without bends, which is the rectilinear image of an orthogonal representation of $G$ with fewer bends than $H$, a contradiction.
Assume now that $g$ is on the external cycle $C_o(G)$ of $G$. If $C_o(\rect{G})$ contains more than four vertices of degree two, we can smooth $v_1$ and apply the same argument as above to contradict the optimality of $H$ (note that, such a smoothing does not violate Condition $(i)$ of Theorem~\ref{th:RN03}). Suppose vice versa that $C_o(\rect{G})$ contains exactly four vertices of degree two (three of them being $v_1$, $v_2$, and $v_3$). In this case, just smoothing $v_1$ violates Condition~$(i)$ of Theorem~\ref{th:RN03}. However, we can smooth $v_1$ and subdivide an edge of $C_o(\rect{G}) \cap C_o(G)$ (such an edge exists since $C_o(G)$ has at least three edges and, by hypothesis and a simple counting argument, at least one of its edges has no bend in $H$). The resulting plane graph $\rect{G''}$ still satisfies the three conditions of Theorem~\ref{th:RN03} and admits a representation $\rect{H''}$ without bends; the representation of which $\rect{H''}$ is the rectilinear image is a bend-minimum orthogonal representation of $G$ with at most two bends per edge.
To see that two bends per edge is worst-case optimal, just consider a bend-minimum representation of the complete graph $K_4$.

\smallskip\noindent{\bf Step~2: Property~\textsf{O2}}. Let $H$ be a bend-minimum orthogonal representation of $G$ that satisfies $\textsf{O1}$ and let $\rect{H}$ be its rectilinear image. The plane underlying graph $\rect{G}$ of $\rect{H}$ satisfies the three conditions of Theorem~\ref{th:RN03}. Rhaman, Nishizeki, and Naznin~\cite[Lemma 3]{DBLP:journals/jgaa/RahmanNN03} prove that, in this case, $\rect{G}$ has an embedding-preserving orthogonal representation $\rect{H'}$ without bends in which every $2$-legged cycle $C$ is either \X-shaped or \D-shaped, where the two poles of the shape are the two leg-vertices of $C$. On the other hand, if $G_\mu$ is an inner P- or R-component, the external cycle $C_o(G_\mu)$ is a $2$-legged cycle of $G$, where the two leg-vertices of $C_o(G_\mu)$ are the poles of $G_\mu$. Hence, the representation $H'$ of $G$ whose rectilinear image is $\rect{H'}$ satisfies \textsf{O2}, as $H'_\mu$ is either \X-shaped or \D-shaped. Also, the bends of $H'$ are the same as in $H$, because the bend vertices of $\rect{H}$ coincide with those of $\rect{H'}$. Hence, $H'$ still satisfies \textsf{O1} and has the minimum number of bends.

\smallskip\noindent{\bf Step~3: Property~\textsf{O3}}. Suppose now that $H$ is a bend-minimum orthogonal representation of $G$ (with $e$ on the external face) that satisfies both \textsf{O1} and \textsf{O2}. More precisely, assume that $H = H'$ is the orthogonal representation obtained in the previous step, where its rectilinear image $\rect{H}$ is computed by the algorithm of Rhaman et al.~\cite{DBLP:journals/jgaa/RahmanNN03}. By a careful analysis of how this algorithm works, we prove that each series gets spirality at most four in $H$ (see Appendix~\ref{se:app-RN}).
\end{proof}

\section{Drawing Algorithm}\label{se:algorithm}
Let $G$ be a biconnected 3-planar graph with a distinguished edge $e$ and let $T$ be the SPQR-tree of $G$ with respect to $e$. Sec.~\ref{sse:inner} gives a linear-time algorithm to compute bend-minimum orthogonal representations of the inner components of $T$. Sec.~\ref{sse:root-child} handles the root child of $T$ to complete a bend-minimum representation with $e$ on the external face and it proves Theorem~\ref{th:main}. Lemma~\ref{le:shapes} allows us to restrict our algorithm to search for a representation satisfying Properties~\textsf{O1-O3}.

%
%
\subsection{Computing Orthogonal Representations for Inner Components}\label{sse:inner}

Let $T$ be the SPQR-tree of $G$ with respect to reference edge~$e$ and let~$\mu$ be an inner node of $T$. A key ingredient of our algorithm is the concept of `equivalent' orthogonal representations of $G_\mu$. Intuitively, two representations of $G_\mu$ are equivalent if one can replace the other in any orthogonal representation of $G$. Similar equivalence concepts have have been used for orthogonal drawings~\cite{DBLP:journals/siamcomp/BattistaLV98,dlp-chvrp-14}.
As we shall prove (see Theorem~\ref{th:substitution}), for planar $3$-graphs a simpler definition of equivalent representations suffices. If $\mu$ is a P- or an R-node, two representations $H_\mu$ and $H'_\mu$ are \emph{equivalent} if they are both \D-shaped or both \X-shaped. If $\mu$ is an inner S-node, $H_\mu$ and $H'_\mu$ are \emph{equivalent} if they have the same spirality.

\begin{restatable}{lemma}{contourpaths}\label{le:contour-paths}
If $H_\mu$ and $H'_\mu$ are two equivalent orthogonal representations of $G_\mu$, the two contour paths of $H_\mu$ have the same turn number as those of $H'_\mu$.
\end{restatable}

Suppose that $H_\mu$ is an inner component of $H$ with poles $u$ and $v$, and let $p_l$ and $p_r$ be the contour paths of $H_\mu$.
\emph{Replacing} $H_\mu$ in $H$ with an equivalent representation $H'_\mu$ means to insert $H'_\mu$ in $H$ in place of $H_\mu$, in such a way that: $(i)$ if $H_\mu$ and $H'_\mu$ are \D-shaped, the contour path $p'$ of $H'_\mu$ for which $t(p')=t(p_l)$ is traversed clockwise from $u$ to $v$ on the external boundary of $H'_\mu$ (as for $p_l$ on the external boundary of $H_\mu$); $(ii)$ in all cases, the external angles of $H'_\mu$ at $u$ and $v$ are the same as in $H_\mu$. This operation may require to mirror $H'_\mu$ (see Fig.~\ref{fi:replacement}). 
The next theorem uses arguments similar to~\cite{DBLP:journals/siamcomp/BattistaLV98}.

\begin{figure}[t]
	\centering
	\subfigure[]{\label{fi:substitution-a}\includegraphics[height=0.18\columnwidth]{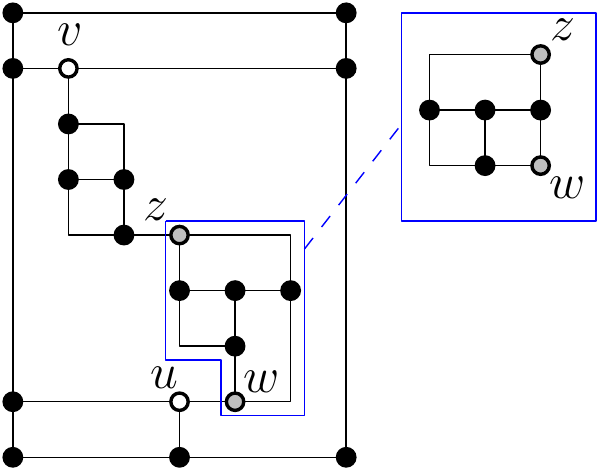}}
	\hfill
	\subfigure[]{\label{fi:substitution-b}\includegraphics[height=0.18\columnwidth]{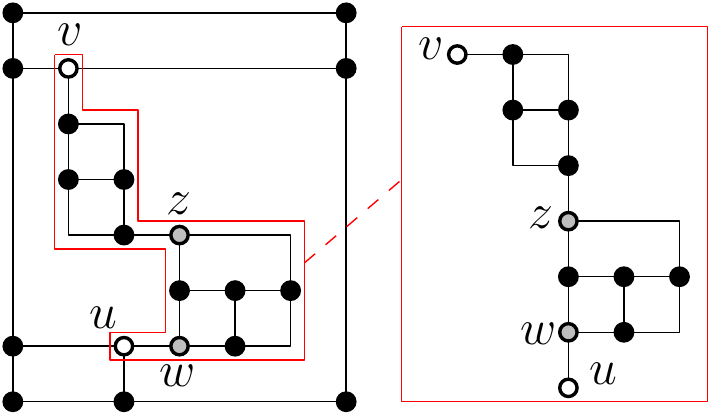}}
	\hfill
	\subfigure[]{\label{fi:substitution-c}\includegraphics[height=0.14\columnwidth]{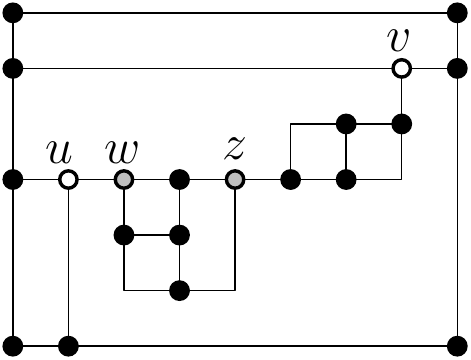}}
	\caption{(a) An orthogonal representation $H$; a D-shaped R-component with poles $\{w,z\}$ and an equivalent representation of it are in the blue frames. (b) A representation obtained from $H$ by replacing the R-component with the equivalent one; a $1$-spiral S-component with poles $\{u,v\}$ and an equivalent one are shown in the red frames. (c) The representation obtained by replacing the S-component with the equivalent one.
	}\label{fi:replacement}
\end{figure}

\begin{restatable}{theorem}{substitution}\label{th:substitution}
Let $H$ be an orthogonal representation of a planar 3-graph $G$ and $H_\mu$ be the restriction of $H$ to~$G_\mu$, where $\mu$ is an inner component of the SPQR-tree $T$ of~$G$ with respect to a reference edge~$e$. Replacing $H_\mu$ in $H$ with an equivalent representation $H'_\mu$ yields a planar orthogonal representation $H'$ of~$G$.
\end{restatable}

We are now ready to describe our drawing algorithm. It is based on a dynamic programming technique that visits bottom-up the SPQR-tree $T$ with respect to the reference edge $e$ of $G$. Based on Lemma~\ref{le:shapes} and Theorem~\ref{th:substitution}, the algorithm stores for each visited node $\mu$ of $T$ a \emph{set of candidate orthogonal representations} of $G_{\mu}$, together with their cost in terms of bends. For a Q-node, the set of candidate orthogonal representations consists of three representations, with 0, 1, and 2 bends, respectively. This suffices by Property~\textsf{O1}. For a P- or an R-node, the set of candidate representations consists of a bend-minimum \D-shaped and a bend-minimum \X-shaped representation. This suffices by Property~\textsf{O2}. For an S-node, the set of candidate representations consists of a bend-minimum representation for each value of spirality $0 \leq k \leq 4$. This suffices by Property~\textsf{O3}. In the following we explain how to compute the set of candidate representations for a node $\mu$ that is a P-, an S-, or an R-node (computing the set of a Q-node is trivial). To achieve overall linear-time complexity, the candidate representations stored at $\mu$ are described incrementally, linking the desired representation in the set of the children of $\mu$ for each virtual edge of $\skel(\mu)$.


\myparagraph{Candidate Representations for a P-node.} By property \textsf{T1} of Lemma~\ref{se:spqr-tree-3-graph}, $\mu$ has two children $\mu_1$ and $\mu_2$, where $\mu_1$ is an S-node and $\mu_2$ is an S-node or a Q-node. The cost of the \X-shaped representation of $\mu$ is the sum of the costs of $\mu_1$ and $\mu_2$ both with spirality one. The cost of the \D-shaped representation of $\mu$ is the minimum between the cost of $\mu_1$ with spirality two and the cost of $\mu_2$ with spirality two. This immediately implies the following.

\begin{lemma}\label{le:P-nodes}
	Let $\mu$ be an inner P-node. There exists an $O(1)$-time algorithm that computes a set of candidate orthogonal representations of $G_\mu$, each having at most two bends per edge.
\end{lemma}

\myparagraph{Candidate Representations for an S-node.} By property \textsf{T3} of Lemma~\ref{se:spqr-tree-3-graph}, $\skel(\mu)$ without its reference edge is a sequence of edges such that the first edge and the last edge are real (they correspond to Q-nodes) and at most one virtual edge, corresponding to either a P- or an R-node, appears between two real edges. Let $c_0$ be the sum of the costs of the cheapest (in terms of bends) orthogonal representations of all P-nodes and R-nodes corresponding to the virtual edges of $\skel(\mu)$. By Property~\textsf{O2}, each of these representations is either \D- or \X-shaped. Let $n_Q$ be the number of edges of $\skel(\mu)$ that correspond to Q-nodes and let $n_D$ be the number of edges of $\skel(\mu)$ that correspond to P- and R-nodes whose cheapest representation is \D-shaped. Obviously, any bend-minimum orthogonal representation of $G_\mu$ satisfying \textsf{O2} has cost at least $c_0$. We have the following.

\begin{restatable}{lemma}{snodespirality}\label{le:s-node-spirality}
An inner S-component admits a bend-minimum orthogonal representation respecting Properties \textsf{O1-O3} and with cost $c_0$ if its spirality $k \leq n_Q+n_D-1$ and with cost $c_0 + k - n_Q-n_D+1$ if $k > n_Q+n_D-1$.
\end{restatable}

Note that the possible presence in $\skel(\mu)$ of virtual edges corresponding to P- and R-nodes whose cheapest representation is \X-shaped does not increase the spirality reachable at cost $c_0$ by the S-node.
Lemma~\ref{le:s-node-spirality} also provides an alternative proof of a known result (\cite[Lemma 5.2]{DBLP:journals/siamcomp/BattistaLV98}), stating that for a planar $3$-graph the number of bends of a bend-minimum $k$-spiral representation of an inner S-component does not decrease when $k$ increases. Moreover, since for an inner S-component $n_Q \geq 2$, a consequence of Lemma~\ref{le:s-node-spirality} is Corollary~\ref{co:k01}. It implies that every bend-minimum $k$-spiral representation of an inner S-component does not require additional bends with respect to the bend-minimum representations of their subcomponents when $k \in \{0,1\}$.

\begin{corollary}\label{co:k01}
For each $k \in \{0,1\}$, every inner S-component admits a bend-minimum orthogonal representation of cost $c_0$ with spirality $k$.	
\end{corollary}

\begin{restatable}{lemma}{Snodes}\label{le:S-nodes}
	Let $\mu$ be an inner S-node and $n_\mu$ be the number of vertices of $\skel(\mu)$. There exists an $O(n_\mu)$-time algorithm that computes a set of candidate orthogonal representations of $G_\mu$, each having at most two bends per edge.
\end{restatable}

\myparagraph{Candidate Representations for an R-node.} If $\mu$ is an R-node, its children are S- or Q-nodes (Property \textsf{T2} of Lemma~\ref{se:spqr-tree-3-graph}). To compute a bend-minimum orthogonal representation of $G_\mu$ that satisfies Properties~\textsf{O1-O3}, we devise a variant of the linear-time algorithm by Rahman, Nakano, and Nishizeki~\cite{DBLP:journals/jgaa/RahmanNN99} that exploits the properties of inner S-components.

\begin{restatable}{lemma}{Rnodes}\label{le:R-nodes}
 Let $\mu$ be an inner R-node and $n_\mu$ be the number of vertices of $\skel(\mu)$. There exists an $O(n_\mu)$-time algorithm that computes a set of candidate orthogonal representations of $G_\mu$, each having at most two bends per edge.
\end{restatable}
\begin{proof}[sketch]
Let $\{u,v\}$ be the poles of $\mu$. Our algorithm works in two steps. First, it computes an \X-shaped orthogonal representation $\tilde{H}^{\x}_\mu$ and a \D-shaped orthogonal representation $\tilde{H}^{\d}_\mu$ of $\tilde{G}_\mu = \skel(\mu) \setminus (u,v)$, with a variant of the recursive algorithm in~\cite{DBLP:journals/jgaa/RahmanNN99}. Then, it computes a bend-minimum \X-shaped representation $H^{\x}_\mu$ and a bend-minimum \D-shaped representation $H^{\d}_\mu$ of $G_\mu$, by replacing each virtual edge $e_S$ in each of $\tilde{H}^{\x}_\mu$ and $\tilde{H}^{\d}_\mu$ with the representation in the set of the corresponding S-node whose spirality equals the number of bends of $e_S$. Every time the algorithm needs to insert a degree-$2$ vertex along an edge of a bad cycle, it adds this vertex on a virtual edge, if such an edge exists. By Corollary~\ref{co:k01}, this vertex does not cause an additional bend in the final representation when the virtual edge is replaced by the corresponding S-component.
\end{proof}

%
%
\subsection{Handling the Root Child Component}\label{sse:root-child}

Let $T$ be the SPQR-tree of $G$ with respect to edge $e=(u,v)$ and let $\mu$ be the root child of~$T$.
Assuming to have already computed the set of candidate representations for the children of~$\mu$, we compute an orthogonal representation $H_\mu$ of $G_\mu$ and a bend-minimum orthogonal representation $H$ of $G$ (with $e$ on the external face) depending on the type of $\mu$.


\myparagraph{Algorithm \textsf{P-root-child}.} Let $\mu$ be a P-node with children $\mu_1$ and $\mu_2$. By Property~\textsf{T1} of Lemma~\ref{se:spqr-tree-3-graph}, both $\mu_1$ and $\mu_2$ are S-nodes. Let $k_1$ ($k_2$) be the maximum spirality of a representation $H_{\mu_1}$ ($H_{\mu_2}$) at the same cost $c_{0,1}$ ($c_{0,2}$) as a $0$-spiral representation. W.l.o.g., let $k_1 \geq k_2$. We have three cases:

\smallskip
\noindent \textsf{Case 1}: $k_1 \geq 4$. Compute a \C-shaped $H_\mu$ by merging a $4$-spiral and a $2$-spiral representation of $\mu_1$ and $\mu_2$, respectively; add $e$ with $0$ bends to get~$H$ (Fig.~\ref{fi:root-child-p-node-a}).

\noindent \textsf{Case 2}: $k_1 = 3$. Compute an \L-shaped $H_\mu$ by merging a $3$-spiral and a $1$-spiral representation of $\mu_1$ and $\mu_2$, respectively; add $e$ with $1$ bend to get~$H$ (Fig.~\ref{fi:root-child-p-node-b}).

\noindent \textsf{Case 3}: $k_1 = 2$ or $k_2 = k_1 = 1$. Compute a \D-shaped $H_\mu$ by merging a $2$-spiral and a $0$-spiral representation of $\mu_1$ and $\mu_2$, respectively; add $e$ with $2$ bends to get~$H$  (Figs.~\ref{fi:root-child-p-node-c}-~\ref{fi:root-child-p-node-d}).


\begin{restatable}{lemma}{Prootchild}\label{le:P-root-child}
\textsf{P-root-child} computes a bend-minimum orthogonal representation of $G$ with $e$ on the external face and at most two bends per edge in $O(1)$~time.
\end{restatable}

\myparagraph{Algorithm \textsf{S-root-child}.} Let $\mu$ be an S-node. if $G_\mu$ starts and ends with one edge, we compute the candidate orthogonal representations of $G_\mu$ as if it were an inner S-node, and we obtain $H$ by adding $e$ with zero bends to the $2$-spiral representation of $G_\mu$ (Fig.~\ref{fi:root-child-s-node-a}). Else, if $G_\mu$ only starts or ends with one edge, we add $e$ to the other end of $G_\mu$, compute the candidate representations of $G_\mu \cup \{e\}$ as if it were an inner S-node, and obtain $G$ by adopting the representation of $G_\mu \cup \{e\}$ with spirality $3$ and by identifying the first and last vertex (Fig.~\ref{fi:root-child-s-node-b}). Finally, if $\skel(\mu) \setminus \{e\}$ starts and ends with an R- or a P-node, we add two copies $e'$, $e''$ of $e$ at the beginning and at the end of $G_\mu$, compute the candidate representations of $G_\mu \cup \{e',e''\}$ as if it were an inner S-node, and obtain $H$ from the representation of $G_\mu \cup \{e',e''\}$ with spirality $4$, by identifying the first and last vertex of $G_\mu \cup \{e',e''\}$ and by smoothing the resulting vertex (Fig.~\ref{fi:root-child-s-node-c}).

\begin{restatable}{lemma}{Srootchild}\label{le:S-root-child}
\textsf{S-root-child} computes a bend-minimum orthogonal representation of $G$ with $e$ on the external face and at most two bends per edge in $O(n_\mu)$~time, where $n_\mu$ is the number of vertices of $\skel(\mu)$.
\end{restatable}


\myparagraph{Algorithm \textsf{R-root-child}.} Let $\mu$ be an R-node and let $\phi_1$ and $\phi_2$ be the two planar embeddings of $\skel(\mu)$ obtained by choosing as external face one of those incident to~$e$. For each $\phi_i$, compute an orthogonal representation $H_i$ of $G$ by: $(i)$ finding a representation $\tilde{H}_i$ of $\skel(\mu)$ (included $e$) with the variant of~\cite{DBLP:journals/jgaa/RahmanNN99} given in the proof of Lemma~\ref{le:R-nodes}, but this time assuming that all the four designated corners of the external face in the initial step must be found; $(ii)$ replacing each virtual edge that bends $k \geq 0$ times in $\tilde{H}_i$ with a minimum-bend $k$-spiral representation of its corresponding S-component. $H$ is the cheapest of $H_1$ and $H_2$. Since the variant of~\cite{DBLP:journals/jgaa/RahmanNN99} applied to $\skel(\mu)$ still causes at most two bends per edge, with the same arguments as in Lemma~\ref{le:R-nodes} we have:

\begin{lemma}\label{le:R-root-child}
\textsf{R-root-child} computes a bend-minimum orthogonal representation of $G$ with $e$ on the external face and at most two bends per edge in $O(n_\mu)$~time, where $n_\mu$ be the number of vertices of $\skel(\mu)$.
\end{lemma}


%
%

\noindent{\bf Proof of Theorem~\ref{th:main}.}
If $G$ is biconnected, Lemmas~\ref{le:P-nodes},~\ref{le:S-nodes},~\ref{le:R-nodes},~\ref{le:P-root-child}$-$\ref{le:R-root-child} yield an $O(n)$-time algorithm that computes a bend-minimum orthogonal representation of $G$ with a distinguished edges $e$ on the external face and at most two bends per edge. Call \textsf{BendMin-RefEdge} this algorithm. An extension of \textsf{BendMin-RefEdge} to a simply-connected graph $G$, which still runs in $O(n)$ time, is easily derivable by exploiting the block-cut-vertex tree of $G$ (see Appendix~\ref{se:app-drawing-alg}). Running \textsf{BendMin-RefEdge} for every possible reference edge, we find in $O(n^2)$ time a bend-minimum orthogonal representation of $G$ over all its planar embeddings. If $v$ is a distinguished vertex of $G$, running \textsf{BendMin-RefEdge} for every edge incident to $v$, we find in $O(n)$ time a bend-minimum orthogonal representation of $G$ with $v$ on the external face (recall that $\deg(v) \leq 3$).
Finally, an orthogonal drawing of $G$ is computed in $O(n)$ time from an orthogonal representation of~$G$~\cite{DBLP:books/ph/BattistaETT99}.
\section{Open Problems}\label{se:conclusions}

We suggest two research directions related to our results: (i) Is there an $O(n)$-time algorithm to compute a bend-minimum orthogonal drawing of a 3-connected planar cubic graph, for every possible choice of the external face? (ii) It is still unknown whether an $O(n)$-time algorithm for the bend-minimization problem in the fixed embedding setting exists~\cite{elt-gd-17}. This problem could be tackled with non-flow based approaches. A positive result in this direction is given in~\cite{DBLP:conf/wg/RahmanN02} for plane 3-graphs.


%

\newpage

\bibliography{bibliography}
\bibliographystyle{splncs04}

\clearpage
\appendix

\makeatletter
\noindent
\rlap{\color[rgb]{0.51,0.50,0.52}\vrule\@width\textwidth\@height1\p@}%
\hspace*{7mm}\fboxsep1.5mm\colorbox[rgb]{1,1,1}{\raisebox{-0.4ex}{%
		\large\selectfont\sffamily\bfseries Appendix}}%
\makeatother

\section{Additional Material for Section~\ref{se:preliminaries}}\label{se:app-prel}
 $G$ is \emph{$1$-connected}, or \emph{simply-connected}, if there is a path between any two vertices. $G$ is \emph{$k$-connected}, for $k \geq 2$, if the removal of $k-1$ vertices leaves the graph $1$-connected. A $2$-connected ($3$-connected) graph is also called \emph{biconnected} (\emph{triconnected}).

A \emph{planar drawing} of $G$ is a geometric representation in the plane such that: $(i)$ each vertex $v \in V(G)$ is drawn as a distinct point $p_v$; $(ii)$ each edge $e=(u,v) \in E(G)$ is drawn as a simple curve connecting $p_u$ and $p_v$; $(iii)$ no two edges intersect in $\Gamma$ except at their common end-vertices (if they are adjacent). A graph is \emph{planar} if it admits a planar drawing. A planar drawing $\Gamma$ of $G$ divides the plane into topologically connected regions, called \emph{faces}. The \emph{external face} of $\Gamma$ is the region of unbounded size; the other faces are \emph{internal}. A \emph{planar embedding} of $G$ is an equivalence class of planar drawings that define the same set of (internal and external) faces, and it can be described by the clockwise sequence of vertices and edges on the boundary of each face plus the choice of the external face. Graph $G$ together a given planar embedding is an \emph{embedded planar graph}, or simply a \emph{plane graph}: If $\Gamma$ is a planar drawing of $G$ whose set of faces is that described by the planar embedding of $G$, we say that $\Gamma$ \emph{preserves} this embedding, or also that $\Gamma$ is an \emph{embedding-preserving drawing} of $G$. 

\spqrtreethreegraph*
\begin{proof}
	We prove the four properties separately.	
	\begin{itemize}
		\item{Proof of \textsf{T1}}. By definition, a P-node $\mu$ has at least two children. Also, since $G$ has no multiple edges, $\mu$ has at most one child that is a Q-node. At the same time, since $\Delta(G) \leq 3$, $\mu$ has neither three children nor a child that is an R-node, as otherwise at least one of its poles would have degree greater than $3$ in $G$. Finally, if $\mu$ is the root child of $T$, its poles coincides with the end-vertices of the reference edge $e$; if a child of $\mu$ is a Q-node, $G$ has multiple edges, a contradiction.
		\item{Proof of \textsf{T2}}. Let $\mu$ be a child of an R-node; if $\mu$ is a P-node or an R-node, the poles of $\mu$ have degree greater than one in $G_\mu$, which implies that these vertices have degree greater than three in $G$, a contradiction.
		\item{Proof of \textsf{T3}}. Let $\mu$ be an inner S-node of $T$ and let $w$ be a pole of $\mu$. The parent of $\mu$ in $T$ is either a P-node or an R-node. If it is a P-node, then the edge incident to $w$ in $\skel(\mu)$ cannot be a virtual edge, as otherwise $w$ has degree at least two in $G_\mu$ and it has at least two edges outside $G_\mu$; this would contradict the fact that $\Delta(G) \leq 3$. If the parent node of $\mu$ is an R-node, then $w$ has exactly three incident edges in the skeleton of this R-node, thus it must have degree one in $G_\mu$. Finally, two virtual edges cannot be incident in $\skel(\mu)$, as they would imply a vertex of degree four in the graph.
	\end{itemize}
\end{proof}

\section{Additional Material for Section~\ref{se:properties-ortho}}\label{se:app-RN}
\shapes*
\begin{proof}
	We prove in three steps the existence of a bend-minimum orthogonal representation $H$ that satisfies~\textsf{O1-O3}. We start by a bend-minimum orthogonal representation of $G$ with $e$ on the external face, and in the first step we prove that it either satisfies \textsf{O1} or it can be locally modified, without changing its planar embedding, so to satisfy \textsf{O1}. In the second step, we prove that from the orthogonal representation obtained in the first step we can derive a new orthogonal representation (still with same embedding) that satisfies \textsf{O2} in addition to \textsf{O1}. Finally, we prove that this last representation also satisfies \textsf{O3}.
	
	\smallskip\noindent{\bf Step~1: Property~\textsf{O1}}. Suppose that $H$ is a bend-minimum orthogonal representation of $G$ with $e$ on the external face and having an edge $g$ (possibly coincident with $e$) with at least three bends. Let $\rect{H}$ be the rectilinear image of $H$, and let $\rect{G}$ be the plane graph underlying $\rect{H}$. Since $\rect{H}$ has no bend, $\rect{G}$ satisfies Conditions~$(i)-(iii)$ of Theorem~\ref{th:RN03}. Denote by $v_1$, $v_2$, and $v_3$ three bend vertices in $\rect{H}$ that correspond to three bends of $g$ in $H$.
	Assume first that $g$ is an internal edge of $G$ (i.e., $g$ does not belong to the external face). Let $\rect{G'}$ be the plane graph obtained from $\rect{G}$ by smoothing $v_1$. We claim that $\rect{G'}$ still satisfies Conditions~$(i)-(iii)$ of Theorem~\ref{th:RN03}. Indeed, if this is not the case, there must be a bad cycle in $\rect{G'}$ that contains both $v_2$ and $v_3$. This is a contradiction, because no bad cycle can contain two vertices of degree two. It follows that there exists an (embedding-preserving) orthogonal representation $\rect{H'}$ of $\rect{G'}$ without bends, which is the rectilinear image of an orthogonal representation of $G$ with fewer bends than $H$, a contradiction.
	Assume now that $g$ is on the external cycle $C_o(G)$ of $G$. If $C_o(\rect{G})$ contains more than four vertices of degree two, then we can smooth vertex $v_1$ and apply the same argument as above to contradict the optimality of $H$ (note that, such a smoothing does not violate Condition $(i)$ of Theorem~\ref{th:RN03}). Suppose vice versa that $C_o(\rect{G})$ contains exactly four vertices of degree two (three of them being $v_1$, $v_2$, and $v_3$). In this case, just smoothing $v_1$ violates Condition~$(i)$ of Theorem~\ref{th:RN03}. However, we can smooth $v_1$ and subdivide an edge of $C_o(\rect{G}) \cap C_o(G)$ (such an edge exists because $C_o(G)$ has at least three edges and, by hypothesis and a simple counting argument, at least one of its edges has no bend in $H$). The resulting plane graph $\rect{G''}$ still satisfies the three conditions of Theorem~\ref{th:RN03} and admits a representation $\rect{H''}$ without bends; the orthogonal representation of which $\rect{H''}$ is the rectilinear image is a bend-minimum orthogonal representation of $G$ with at most two bends per edge.
	To see that two bends per edge is worst case optimal, just consider the complete graph $K_4$ on four vertices. Every planar embedding of $K_4$ has three edges on $C_o(K_4)$. By Condition~$(i)$ of Theorem~\ref{th:RN03}, a bend-minimum orthogonal representation of $K_4$ has four bends on the external face and thus two of them are on the same edge.

	\smallskip\noindent{\bf Step~2: Property~\textsf{O2}}. Let $H$ be a bend-minimum orthogonal representation of $G$ that satisfies $\textsf{O1}$ and let $\rect{H}$ be its rectilinear image. The plane underlying graph $\rect{G}$ of $\rect{H}$ satisfies the three conditions of Theorem~\ref{th:RN03}. Rhaman, Nishizeki, and Naznin~\cite[Lemma 3]{DBLP:journals/jgaa/RahmanNN03} prove that, in this case, $\rect{G}$ has an embedding-preserving orthogonal representation $\rect{H'}$ without bends in which every $2$-legged cycle $C$ is either \X-shaped or \D-shaped, where the two poles of the shape are the two leg-vertices of $C$. On the other hand, if $G_\mu$ is an inner P-component or an inner R-component, the external cycle $C_o(G_\mu)$ is a $2$-legged cycle of $G$, where the two leg-vertices of $C_o(G_\mu)$ are the poles of $G_\mu$. Indeed, $C_o(G_\mu)$ is a simple cycle and each pole has exactly one incident edge not belonging to $G_\mu$. It follows that, the orthogonal representation $H'$ of $G$ whose rectilinear image is $\rect{H'}$ satisfies \textsf{O2}, as $H'_\mu$ is either \X-shaped or \D-shaped. Also note that the bends of $H'$ are the same as in $H$, because the bend vertices of $\rect{H}$ coincide with those of $\rect{H'}$. Hence, $H'$ still satisfies \textsf{O1} and has the minimum number of bends.
	
	\smallskip\noindent{\bf Step~3: Property~\textsf{O3}}. Suppose now that $H$ is a bend-minimum orthogonal representation of $G$ (with $e$ on the external face) that satisfies both \textsf{O1} and \textsf{O2}. More precisely, assume that $H = H'$ is the orthogonal representation obtained in the previous step, where its rectilinear image $\rect{H}$ is computed by the algorithm of Rhaman et al.~\cite{DBLP:journals/jgaa/RahmanNN03}, which we simply call \RN. 
	This algorithm works as follows (see also Fig.~\ref{fi:RN03}). In the first step, it arbitrarily designates four degree-2 vertices $x, y, w, z$ on the external face. A 2-legged cycle (resp. 3-legged cycle) of the graph is \emph{bad} with respect to these four vertices if it does not contain at least two (resp. one) of them; a bad cycle $C$ is \emph{maximal} if it is not contained in $G(C')$ for any other bad cycle $C'$. The algorithm finds every maximal bad cycle $C$ (the maximal bad cycles are independent of each other) and it collapses $G(C)$ into a supernode $v_C$. Then it computes a rectangular representation $R$ of the resulting coarser plane graph (i.e., a representation with all rectangular faces) where each of $x, y, w, z$ (or a supernode containing it) is an external corner. Such a representation exists because the graph satisfies a characterization of Thomassen~\cite{Th84}. In the next steps, for each supernode $v_C$, \RN recursively applies the same approach to compute an orthogonal representation of $G(C)$; if $C$ is 2-legged (resp. 3-legged), then two (resp. three) designated vertices coincide with the leg-vertices of $C$. The representation of each supernode is then ``plugged'' into $R$.
	
	Suppose now that $H_\mu$ is an inner S-component of $H$ and let $u$ and $v$ be its poles. Let $B_1, \dots B_h$ be the biconnected components of $G_\mu$ that are not single edges. We call each $B_i$ a \emph{subcomponent} of $G_\mu$ (if $G_\mu$ is a sequence of edges, it has no subcomponents). Consider a generic step of \RN, in which it has to draw $G(C)$, for some cycle $C$ (possibly the external cycle of $G$) such that $G_\mu \subseteq G(C)$. We distinguish between three cases.
	
	\smallskip\noindent{\textsf{Case 1 - }} $G_\mu$ is not contained in any maximal bad cycle. If all the edges of $G_\mu$ are internal edges of $G(C)$, the external cycle $C_o(B_i)$ of each $B_i$ is a maximal bad 2-legged cycle (as it contains no designated vertices). In this case each $G(C_0(B_i))$ is collapsed into a supernode, and in the rectangular drawing $R$ of the resulting graph all the degree-2 vertices and supernodes of the series will belong to the same side of a rectangular face. Thus, when all subcomponents of $H_\mu$ are drawn and plugged into $R$, $H_\mu$ gets spirality zero. If $G_\mu$ has some edges on the external face, some $C_o(B_i)$ might not be a maximal bad cycle (in which case it contains at least two designated vertices); in this case the spirality of $H_\mu$ will be equal to the number of designated vertices on its external edges, which is at most four.
	
	\smallskip\noindent{\textsf{Case 2 - }} $G_\mu$ is contained in a maximal bad cycle $C'$ that passes through both $u$ and $v$. In this case, $G(C')$ is collapsed into a supernode $v_{C'}$ before the computation of a rectangular drawing $R$. The two legs of $C'$ that are incident to $u$ and $v$ will form at $v_{C'}$ in $R$ either two flat angles or a right angle: In the former case, $H_\mu$ will have spirality zero, while in the latter case it will have spirality one.
	
	\smallskip\noindent{\textsf{Case 3 - }} $G_\mu$ is contained in a maximal bad cycle $C'$ that does not passes through both $u$ and $v$. In this case, $G(C')$ is still collapsed into a supernode $v_{C'}$ before the computation of a rectangular drawing $R$, and in one of the subsequent steps of the recursive algorithm on $G(C')$ we will fall in Case~1 or in Case~2.
\end{proof}

Figure~\ref{fi:RN03} shows an example of how \RN works. The input plane graph is in Fig.~\ref{fi:RN03-a}, and it is the same graph $G$ of Fig.~\ref{fi:spqr-tree-a}, with the addition of some degree-2 vertices (small squares), needed to satisfy the properties of Theorem~\ref{th:RN03}. 
The external face of $G$ contains exactly four degree-2 vertices, which are assumed to be the four designated vertices in the first step of \RN.
In the figure, the bad cycles with respect to the designated vertices are highlighted in red; the two cycles with thicker boundaries are maximal, and therefore they are collapsed as shown in Fig.~\ref{fi:RN03-b}. Note that, one of the two maximal cycles includes a designated vertex; once this cycle is collapsed, the corresponding supernode becomes the new designated vertex. Figure~\ref{fi:RN03-c} depicts a rectangular representation of the graph in Fig.~\ref{fi:RN03-b}, and it also shows the representations of the subgraphs in the supernodes, computed in the recursive procedure of \RN; these representations are plugged in the rectangular representation, in place of the supernodes, yielding the final representation of Fig.~\ref{fi:RN03-d}.     

\begin{figure}[!h]
	\centering
	\subfigure[]{\label{fi:RN03-a}\includegraphics[width=0.26\columnwidth]{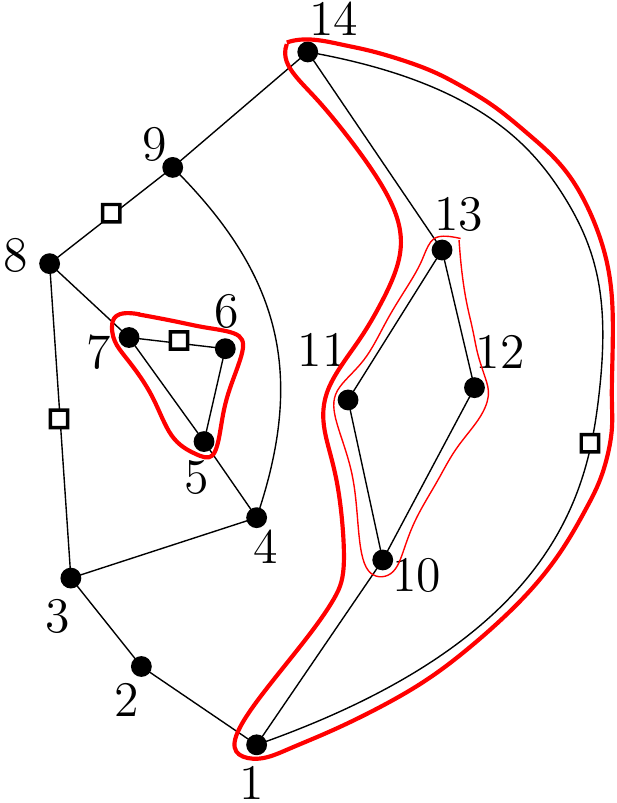}}
	\hfill
	\subfigure[]{\label{fi:RN03-b}\includegraphics[width=0.22\columnwidth]{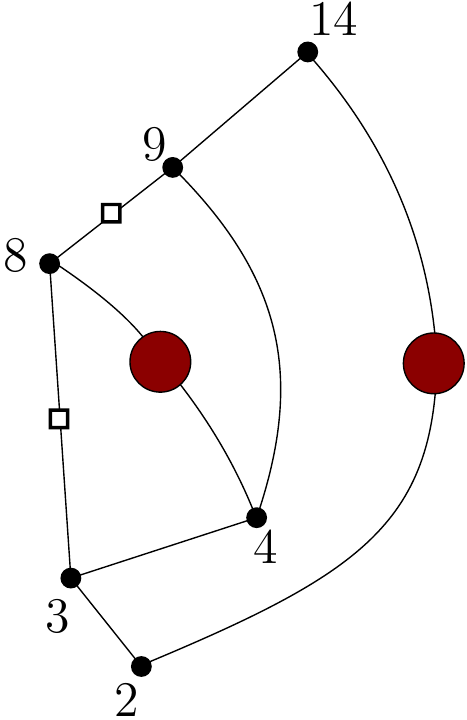}}
	\hfill
	\subfigure[]{\label{fi:RN03-c}\includegraphics[width=0.29\columnwidth]{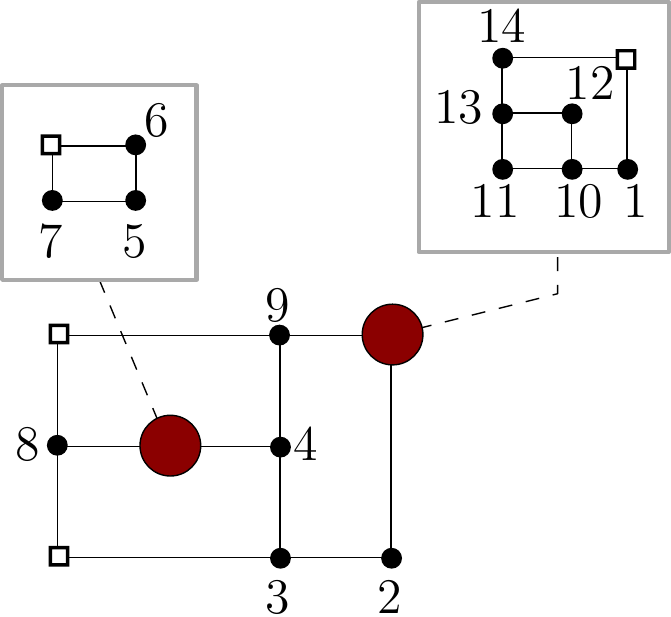}}
	\hfill
	\subfigure[]{\label{fi:RN03-d}\includegraphics[width=0.29\columnwidth]{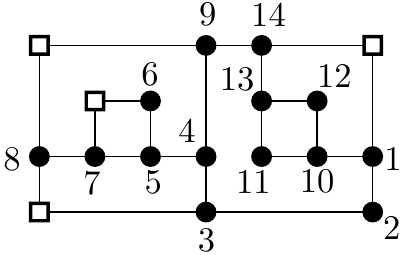}}
	\caption{An illustration of the algorithm \RN, described by Rahaman, Nishizeki, and Naznin~\cite{DBLP:journals/jgaa/RahmanNN03}.}\label{fi:RN03}
\end{figure}

\section{Additional Material for Section~\ref{se:algorithm}}\label{se:app-drawing-alg}
\contourpaths*
\begin{proof}
	The proof is trivial if $\mu$ is a P-node or an R-node since, in order to be equivalent, $H_\mu$ and $H'_\mu$ must be both \D-shaped or \X-shaped, which, by definition, implies that their contour paths have the same turn numbers. If $\mu$ is an S-node, then $H_\mu$ and $H'_\mu$ have the same spirality~$k$ and, by Lemma~\ref{le:k-spiral}, their contour paths have the same turn number $k$ as any path from one pole of $\mu$ to the other pole.
\end{proof}

\substitution*
\begin{proof}
The statement easily follows from Lemma~\ref{le:contour-paths} and from a characterization of orthogonal representations first stated in~\cite{DBLP:journals/siamcomp/VijayanW85}. Let $G$ be a biconnected embedded planar graph and let $\phi$ be a function that assigns: (i) a value in $\{90^\textrm{o},180^\textrm{o},270^\textrm{o}\}$ to each pair of consecutive edges in the circular order around each vertex $v$ of~$G$ and (ii) a sequence of left and right turns to each edge $e=(u,v)$ of~$G$. Denote by $t(f)$ the \emph{turn number} of a face $f$, that is, the difference between the number of right and the number of left turns encountered while clockwise traversing the border of~$f$, where a $180^\textrm{o}$ angle counts as zero, a $90^\textrm{o}$ angle (resp., $270^\textrm{o}$ angle) counts as one if $f$ is an internal face (resp., if $f$ is the external face), and a $270^\textrm{o}$ angle (resp., $90^\textrm{o}$ angle) counts as $-1$ if $f$ is an internal face (resp., if $f$ is the external face). Then $\phi$ corresponds to an orthogonal representation $H_\phi$ of $G$ if and only if the turn number $t(f)$ of each face $f$ is four. 

Hence, an orthogonal representation $H$ of $G$ induces an assignment $\phi$ satisfying the above mentioned property. Let $\phi_\mu$ be the restriction of $\phi$ to the internal faces of $H_\mu$. Let $H'_\mu$ be an orthogonal representation of $G_\mu$ equivalent to $H_\mu$ and let $\phi'_\mu$ be the assignment corresponding to $H'_\mu$. We now prove that replacing $H_\mu$ with $H'_\mu$ yields a new orthogonal representation $H'$ by showing that the corresponding function $\phi'$ satisfies the above characterization. 
In the following, assuming that $u$ and $v$ are the two poles of $H_\mu$ and that $p_l$ and $p_r$ are the two contour paths of $H_\mu$, we call $p'_l$ the contour path of $H'_\mu$ such that $t(p'_l)=t(p_l)$ and $p'_r$ the contour path of $H'_\mu$ such that $t(p'_r)=t(p_r)$ (if a mirroring of $H'_\mu$ is needed in the replacement, the two contour paths of $H'_\mu$ are renamed). Also, let $f_l$ be the face of $H$ to the left of $p_l$ while moving along $p_l$ from $u$ to $v$, and let $f'_l$ be the face of $H'$ to the left of $p'_l$ while moving along $p'_l$ from $u$ to $v$. The faces $f_r$ and $f'_r$ are defined symmetrically.     
Assignment $\phi'$ is such that: (a) Each face $f$ of $H'$ whose boundary is exclusively composed of edges not belonging to $G_\mu$ has the same shape as in $H$ and for the angles and edges of $f$ we have $\phi' = \phi$, hence $t(f)=4$. (b) Each face $f$ whose boundary is exclusively composed of edges of $G_\mu$ has the same shape as in $H'_\mu$ and for the angles and edges of $f$ we have $\phi' = \phi'_\mu$, hence $t(f)=4$. 
(c) The remaining two faces of $H'$ are $f'_l$ and $f'_r$. The boundary of $f_l$ is composed of path $p_l$ plus another path $p$. By construction, the boundary of $f'_l$ is composed of path $p'_l$ plus $p$. Also, again by construction, the two angles at $u$ and $v$ in $f_l$ are the same as the two angles at $u$ and $v$ in $f'_l$. Since $t(p_l)=t(p'_r)$ the turn number of $f'_l$ in $H'$ equals the turn number of $f_l$ in $H$. The same for $f'_r$.
\end{proof}


\snodespirality*
\begin{proof}
	The proof is by induction on the number of children of $\mu$ that are not Q-nodes. Suppose that $\mu$ has only Q-node children ($n_D = 0$). It is trivial that $G_\mu$, which is a path, can be drawn with cost zero and with spirality in $[0,\dots,n_Q-1]$, while has increasing costs for higher values of spirality. For the inductive case, note that inserting an \X-shaped child in between two Q-nodes does not increase nor decrease the spirality of an orthogonal drawing of $G_\mu$. Instead, a \D-shaped child inserted in between two Q-nodes can be used as if it were an additional Q-node to increase the spirality of one unit without additional costs.
\end{proof}

\Snodes*
\begin{proof}
	By virtue of Lemma~\ref{le:s-node-spirality} we can sum up the cheapest costs of the representations of all P- and R-node children of $\mu$ to obtain the cost $c_0$ of a bend-minimum orthogonal representation of $\mu$ with spirality in $[0,\dots,n_Q+n_D-1]$. If $n_Q+n_D-1 \geq 4$ we are done. Otherwise, by Property~\textsf{O2} of Lemma~\ref{le:shapes}, we can optimally increase the spirality of $G_\mu$ by inserting bends into the Q-nodes of $\skel(\mu)$. Since $n_Q \geq 2$ and the needed extra bends are at most three (because $k \leq 4$), if we evenly distribute the extra bends among the real edges of $\skel(\mu)$, we end up with at most two bends per edge, satisfying Property~\textsf{O1} of Lemma~\ref{le:shapes}.
\end{proof}

\Rnodes*
\begin{proof}
We first briefly recall the algorithm in~\cite{DBLP:journals/jgaa/RahmanNN99}, which we call \RNN and which is conceptually similar to \RN (described in the proof of Lemma~\ref{le:shapes}). \RNN takes as input an embedded planar triconnected cubic graph $\hat{G}$ and computes an embedding-preserving bend-minimum orthogonal representation of $\hat{G}$. It initially inserts four dummy vertices $x, y, w, z$ of degree two on $C_o(\hat{G})$, by suitably subdividing some external edges; these vertices act as the four designated vertices of \RN. Note that, since $\hat{G}$ is triconnected, each maximal bad cycle $C$ with respect to  $x, y, w, z$ is a $3$-legged cycle. For each such cycle $C$, the algorithm collapses $\hat{G}(C)$ into a supernode $v_C$, thus obtaining a coarser graph, which admits a rectangular representation $R$~\cite{Th84}.  In the successive steps, for each supernode $v_C$, \RNN recursively applies the same approach to compute an orthogonal representation of $\hat{G(C)}$, where three of the four designated vertices coincide with the leg-vertices of $C$. The representation of~$\hat{G(C)}$ is then plugged into $R$. All designated vertices added throughout the algorithm are bends of the final representation. Crucial to the bend-minimization process of \RNN is the insertion of the designated vertices along edges that are shared by more than one bad cycles (if any). For example, let $C_1$ and $C_2$ be two bad cycles such that $C_1$ is maximal, $C_2$ belongs to $\hat{G}(C_1)$, and $C_1$ and $C_2$ share a path $p$; since both $C_1$ and $C_2$ need a bend, \RNN inserts a designated vertex along $p$ to save bends. At any step of the recursion, a \emph{red} edge is an edge for which placing a designated vertex along it leads to a sub-optimal solution; the remaining edges are \emph{green}. As it is proven in~\cite{DBLP:journals/jgaa/RahmanNN99}, every bad cycle has at least one green edge. A bad cycle $C$ is a {\em corner cycle} if it has at least one green edge on the external face and there is no other bad cycle inside $C$ having this property. In order to save bends when placing the four designated vertices on $C_o(\hat{G})$, \RNN gives preference to the green edges of the corner cycles, if they exist. We remark that, algorithm \RNN produces orthogonal drawings with at most one bend per edge, with the possible exception of one edge in the outer face, which is bent twice if the external boundary is a 3-cycle.	

\smallskip We now describe our algorithm, called \R. Let $\{u,v\}$ be the poles of $\mu$. \R consists of two steps. In the first step it computes an \X-shaped orthogonal representation $\tilde{H}^{\x}_\mu$ and a \D-shaped orthogonal representation $\tilde{H}^{\d}_\mu$ of $\tilde{G}_\mu = \skel(\mu) \setminus (u,v)$, using a variant of \RNN. In the second step, a bend-minimum \X-shaped orthogonal representation $H^{\x}_\mu$ and a bend-minimum \D-shaped orthogonal representation $H^{\d}_\mu$ of $G_\mu$ are constructed, by replacing each virtual edge $e_S$ in each of $\tilde{H}^{\x}_\mu$ and $\tilde{H}^{\d}_\mu$ with the representation in the set of the corresponding S-node whose spirality equals the number of bends of $e_S$. 
Every time the algorithm needs to insert a designated vertex that subdivides an edge of the graph (either to break a bad cycle or to guarantee four external corners in the initial step), it adds this vertex on a virtual edge, if such an edge exists. Indeed, by Corollary~\ref{co:k01}, this vertex does not cause an additional bend in the final representation when the virtual edge is replaced by the corresponding S-component. To find $\tilde{H}^{\x}_\mu$, where $u$ and $v$ are two of the four designated vertices on the external face of $\tilde{G}_\mu$, \R has to find a further designated vertex on each of the two contour paths $p_l$ and $p_r$ of $\tilde{G}_\mu$ from $u$ to $v$. To do this, \R first computes the corner cycles as in \RNN, assuming that $u$ and $v$ are vertices of degree three, that is, attached with an additional leg to the rest of the graph. Since $\skel(\mu) = \tilde{G}_\mu \cup (u,v)$ is cubic and triconnected, the set of corner cycles computed in this way equals the set of corner cycles computed on $\skel(\mu)$, for any of its two possible planar embeddings, minus those that involve $(u,v)$, if any. 
Edges on each of the two contour paths $p \in \{p_l,p_r\}$ are classified as follows: a virtual edge of a corner cycle is \emph{free-\&-useful} (`free' because a bend on such an edge does not correspond to a bend in the final orthogonal representation; `useful' because it also satisfies condition $(iii)$ of Theorem~\ref{th:RN03} for some $3$-legged cycle); a virtual edge that does not belong to a corner cycle is \emph{free-\&-useless}; a (real) edge of a corner cycle is \emph{costly-\&-useful} (`costly' because a bend on such edge is an actual bend in the final orthogonal representation); any other real edge is \emph{costly-\&-useless}. Note that if the edge $e_{l}$ of $p_l$ incident to $u$ (resp. $v$) is useful, also the edge $e_{r}$ of $p_r$ incident to $u$ (resp. $v$) is useful. However, choosing one between $e_l$ and $e_r$ is enough to satisfy condition $(iii)$ of Theorem~\ref{th:RN03} for their common $3$-legged cycles. Hence, choosing $e_l$ will transform $e_r$ into a useless edge and vice versa. 
The two designated vertices on $p_l$ and $p_r$ are chosen in such a way to minimize the sum of the number of useless and the number of the costly edges, which implies the minimization of the number of bends introduced in the final orthogonal drawing of $G_\mu$. Once the four designated vertices are chosen, \R procedes recursively as \RNN. However, each time a new designated vertex has to be added to break a bad cycle $C$, the edge of $C$ along which this vertex is added is chosen according to the following priority: $(i)$ a virtual green edge; $(ii)$ a virtual red edge; $(iii)$ a (real) green edge; $(iv)$ any other real edge. 

The computation of $\tilde{H}^{\d}_\mu$ is similar to that of $\tilde{H}^{\x}_\mu$, with the difference that the two designated vertices on the external face of $\tilde{G}_\mu$, in addition to $u$ and $v$, must be inserted both on $p_l$ or both on $p_r$. 
Further, if a virtual edge $e_{\mu'}$ on the external face of $\tilde{G}_\mu$ corresponds to a child $\mu'$ of $\mu$ that admits an orthogonal drawing with spirality two at the same cost as the orthogonal drawing with spirality zero, then we call such an edge \emph{double-free}, since it can host both the designated vertices without additional costs, and allow it to be chosen two times (only one of the two choices can be useful). Two different computations are performed, one for each contour path, and only the cheapest in terms of bends is considered. 

Concerning the time complexity, \R can be implemented to run in $O(n_\mu)$. Indeed, the first step of \R is a simple modification of \RNN, and thus it works in $O(n_\mu)$ time. The orthogonal representations $H^{\x}_\mu$ and $H^{\d}_\mu$ are then described in a succinct way, by linking the desired representations of the S-component for each virtual edge.  
\end{proof}

\begin{figure}[t]
	\centering
	\subfigure[]{\label{fi:root-child-p-node-a}\includegraphics[width=0.38\columnwidth]{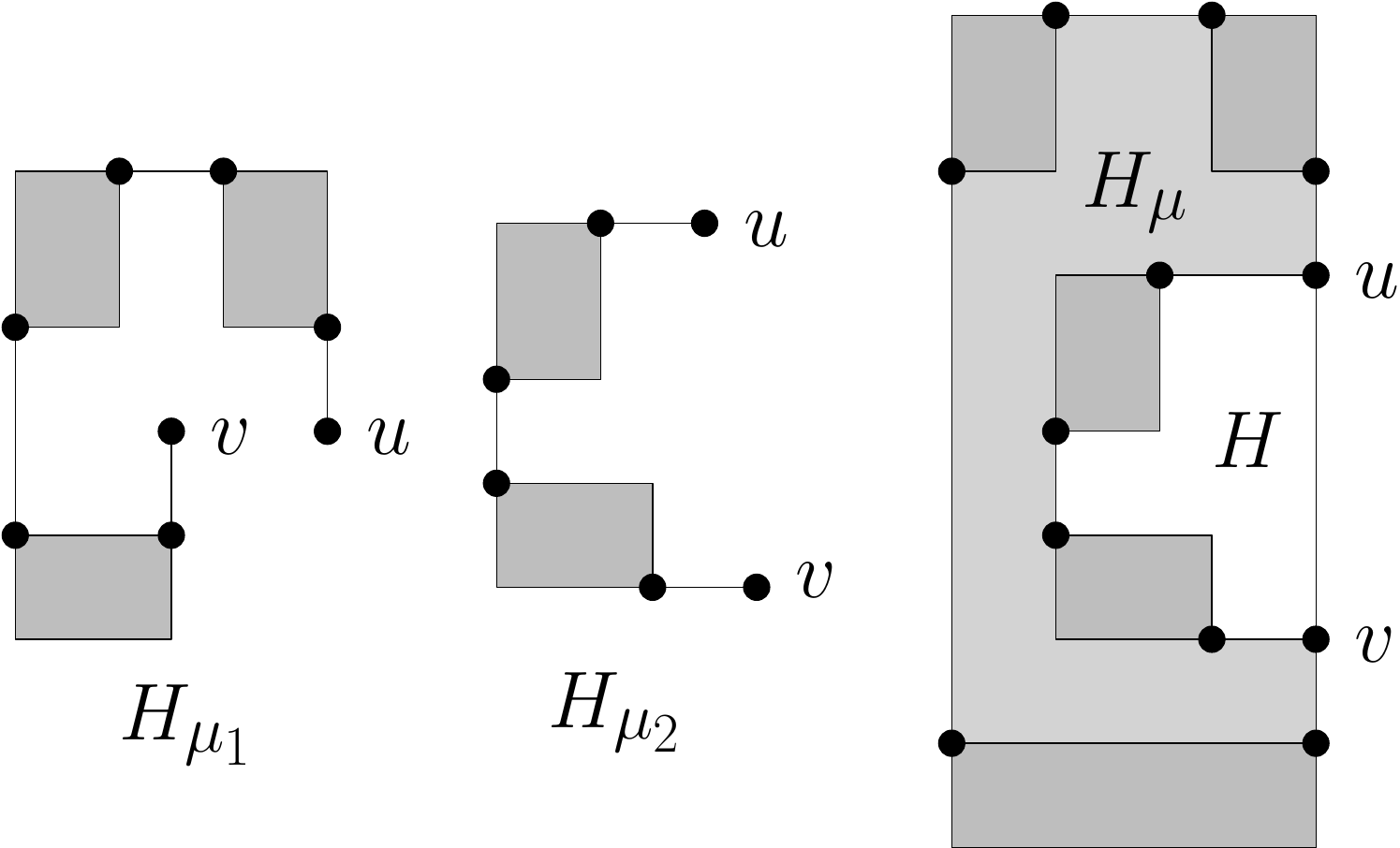}}
	\hfill
	\subfigure[]{\label{fi:root-child-p-node-b}\includegraphics[width=0.38\columnwidth]{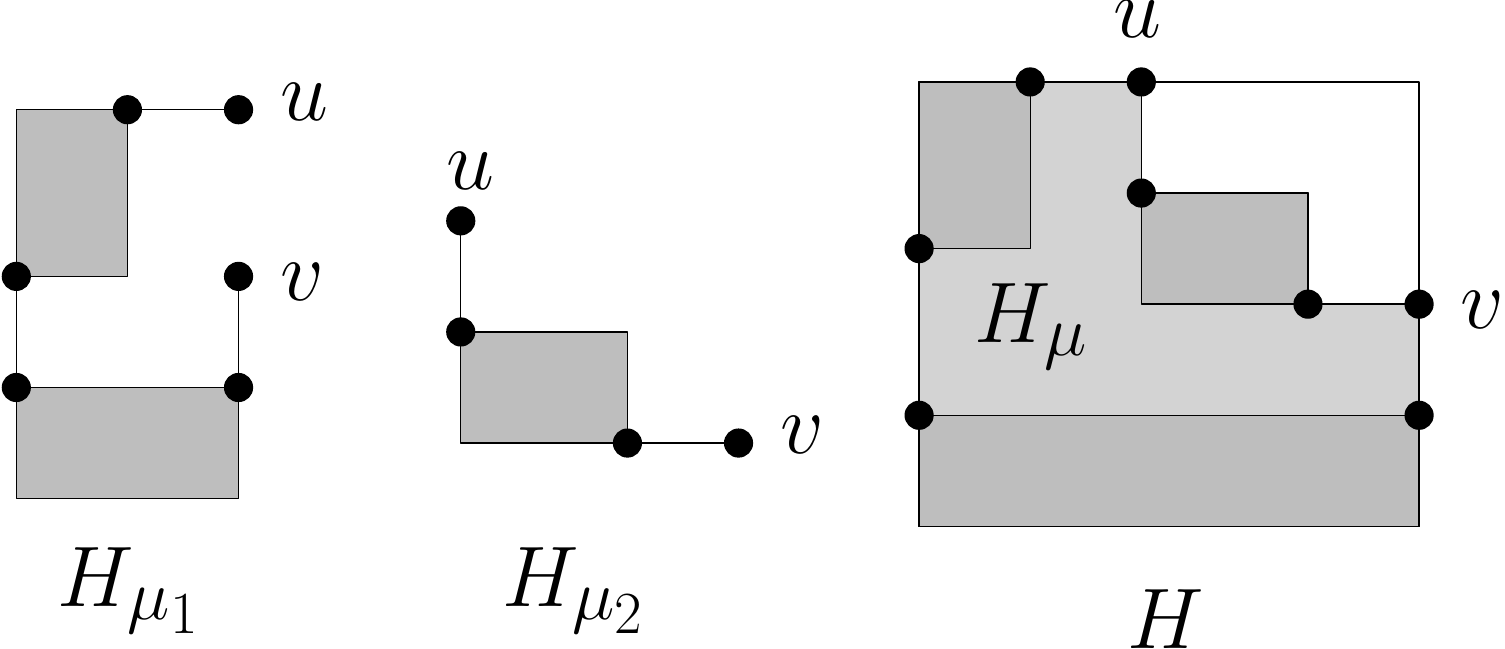}}
	\hfill
	\subfigure[]{\label{fi:root-child-p-node-c}\includegraphics[width=0.38\columnwidth]{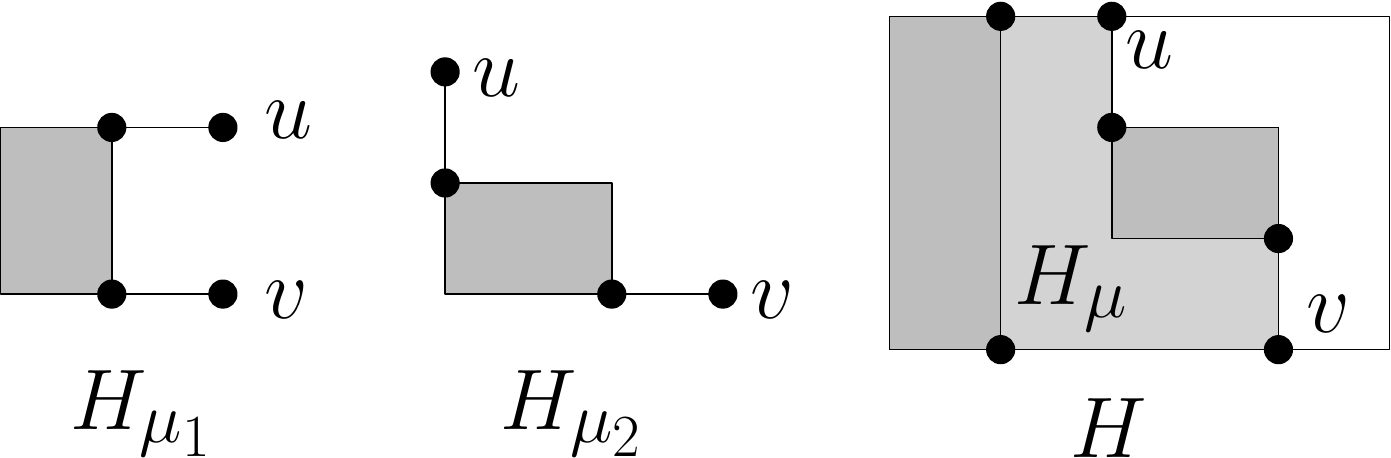}}
	\hfill
	\subfigure[]{\label{fi:root-child-p-node-d}\includegraphics[width=0.33\columnwidth]{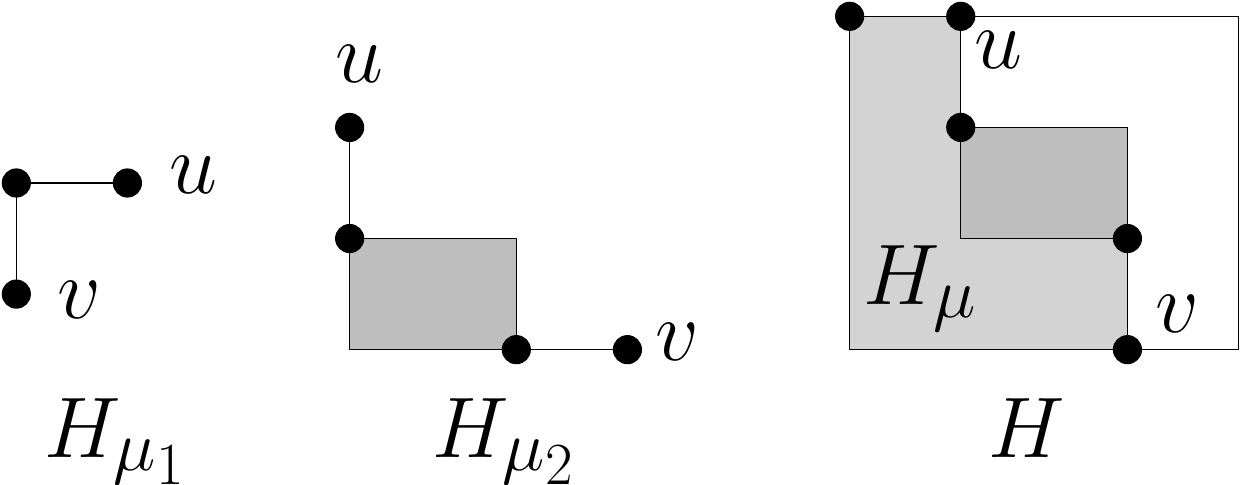}}\\
	\hfill
	\subfigure[]{\label{fi:root-child-s-node-a}\includegraphics[width=0.18\columnwidth]{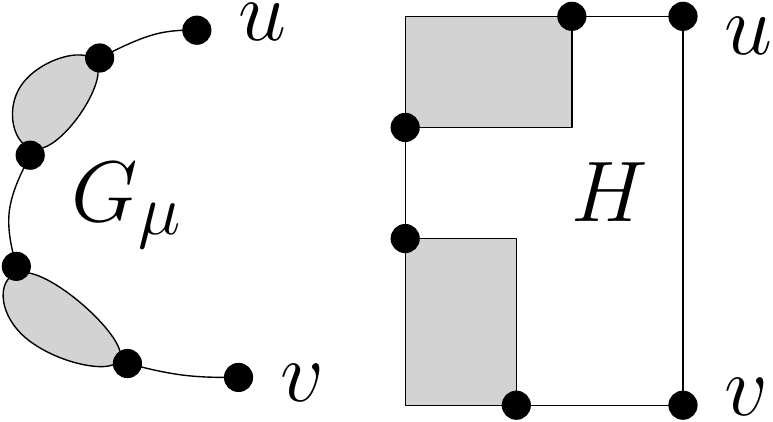}}
	\hfill
	\subfigure[]{\label{fi:root-child-s-node-b}\includegraphics[width=0.36\columnwidth]{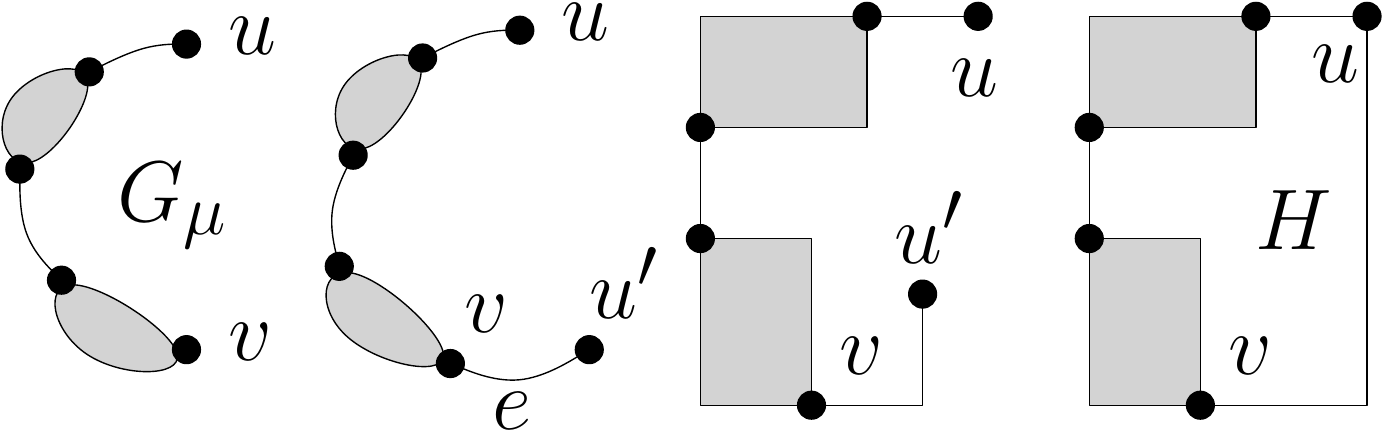}}
	\hfill
	\subfigure[]{\label{fi:root-child-s-node-c}\includegraphics[width=0.36\columnwidth]{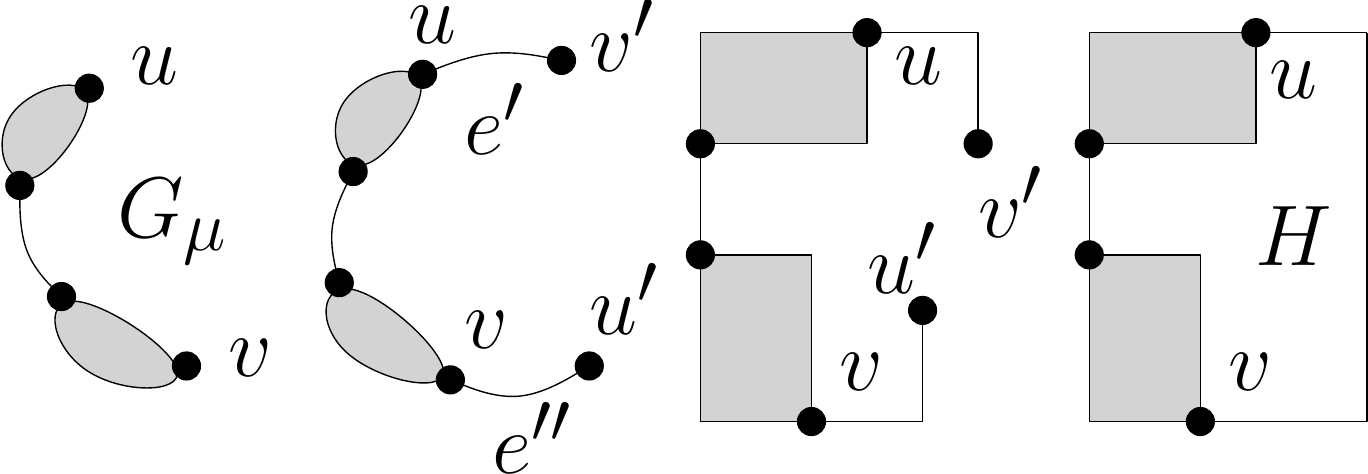}}
	
	\caption{Examples of: (a)-(d) a P-node root child; (e)-(f) an S-node root child.}\label{fi:root-child-p-node}
\end{figure}

\Prootchild*
\begin{proof}
	Observe that any orthogonal representation of $\mu$ satisfies the property that the spirality of $\mu_1$ plus the number of bends of the edge $(u,v)$ is at least four. This property, that we call the \emph{cycle-property}, is due to the fact that $\mu_1$ and $(u,v)$ have to close a cycle containing the drawing of $\mu_2$. 
	Further, since by hypothesis $c_{0,1}$ and $c_{0,2}$ are the minimum number of bends for an orthogonal drawing of $G_{\mu_1}$ and $G_{\mu_2}$, respectively, $G$ does not admit an orthogonal representation with fewer than $c_{0,1}+c_{0,2}$ bends. We consider the costs of the solutions in the four different cases. \textsf{Case 1}: the cost is $c_{0,1} + c_{0,2}$ if $k_2 \geq 2$, which is the minimum possible. Otherwise, by Corollary~\ref{co:k01} we have $k_2 = 1$ and the cost is $c_{0,1} + c_{0,2} + 1$. This is again the minimum, because drawing $\mu_2$ with spirality $1$ would save a bend for $\mu_2$, but would force $(u,v)$ to have one bend. \textsf{Case 2}: the cost is $c_{0,1} + c_{0,2} + 1$, which is the minimum since the cycle-property forces one between $\mu_1$ or $(u,v)$ to host a bend. The same consideration proves that the cost is minimum also for case \textsf{Case 3}, where the costs are $c_{0,1} + c_{0,2} + 2$ and $c_{0,1} + c_{0,2} + 3$, respectively. The obtained orthogonal representation trivially satisfies Properties~\textsf{O1}--\textsf{O3}, since they are satisfied by hypothesys by the candidate representations of $\mu_1$ and $\mu_2$ and $(u,v)$ has at maximum two bends.
\end{proof}

\Srootchild*
\begin{proof}
	The fact that the number of bends is minimum descends from the evidence that it is always convenient to increase the spirality of a series, as it only costs one unit for each unit of increased spirality (Lemma~\ref{le:s-node-spirality}). The obtained orthogonal representation satisfies Properties~\textsf{O1}--\textsf{O3}. In fact, no edge receives more than two bends. In the first case $(u,v)$ is added without bends. In the second case, $G_\mu \cup \{e\}$ is a series with at least two Q-nodes (including $e$), which guarantees that it admits an orthogonal drawing with spirality $1$ without bending the Q-nodes and of spirality three with at most two bends on such Q-nodes. In the third case, $G_\mu \cup \{e',e''\}$ is a series with at least three Q-nodes (including $e'$ and $e''$). Hence, an orthogonal drawing with spirality $4$ has at most two bends on these Q-nodes. 
\end{proof}

\myparagraph{Extension to simply connected graphs}. 
Let $G$ be a simply-connected $3$-planar graph. We exploit the block-cut-vertex tree (or BC-tree) $\cal T$ of the simply-connected 3-planar graph $G$, where each node of $\cal T$ is either a block or a cut-vertex and cut-vertices are adjacent to the blocks they belong to. Call \emph{trivial blocks} those composed of a single edge and \emph{full blocks} the remaining. Since $\Delta(G) \leq 3$, full blocks are only adjacent to trivial blocks. Also, cut-vertices of degree three are adjacent to three trivial blocks. All trivial blocks admit a drawing with zero bends as straight edges.
Let $e$ be an edge that is chosen to be on the external face. Root $\cal T$ at the block $B_e$ containing $e$ and compute in $O(n_e)$ time, where $n_e$ is the number of vertices of $B_e$, a bend-minimum orthogonal representation $H_e$ of $B_e$ with $e$ on the external face as described in Sections~\ref{sse:inner} and~\ref{sse:root-child}. Build an optimal orthogonal representation $H$ of $G$ by recursively attaching pieces to the initial orthogonal representation $H_e$. Let $v$ be a cut-vertex of the current orthogonal representation $H$ and let $B_v$ be a child block attached to $v$. Compute an optimal orthogonal representation $H_v$ of $B_v$ with $v$ on the external face and attach it to $v$. Since $\deg(v) \leq 2$ in $B_v$, in order to have $v$ on the external face it suffices to have one of its two incident edges on the external face. Hence, the required orthogonal representation $H_v$ can be computed in $O(n_v)$ time, where $n_v$ is the number of vertices of~$B_v$.

\end{document}